\documentclass[12pt]{article}
{\begin{Sbox}\begin{minipage}}%
{\end{minipage}\end{Sbox}\fbox{\TheSbox}}%

\usepackage[psamsfonts]{amssymb}
\usepackage{amsmath, amsthm, anysize, enumerate, color, url}
\usepackage{graphicx}
\usepackage{epstopdf}
\usepackage{epsfig}
\usepackage{subfigure}
\usepackage[ruled,linesnumbered]{algorithm2e}
\usepackage{booktabs}
\usepackage{bm}
\usepackage{xcolor}

\definecolor{yancolor}{rgb}{0.06, 0.3, 0.57}
\definecolor{lightblue}{rgb}{0.66, 0.85, 0.95}
\definecolor{PrussianBlue}{RGB}{0, 49, 83}
\definecolor{darkred}{rgb}{0.5, 0.0, 0.0}

\usepackage[sectionbib]{natbib}

\usepackage{threeparttable}

\usepackage[colorlinks, breaklinks,
                   linkcolor=red,
                  citecolor=blue
                  ]{hyperref}

\usepackage{breakurl}

\newtheorem{theorem}{Theorem} 
\newtheorem{corollary}{Corollary}

\newtheorem{proposition}{Proposition} 

\theoremstyle{definition}
\newtheorem{remark}{Remark}  

\newtheorem{assumption}{Assumption}

\newcommand{\E}{\mathbb{E}}

\newcommand{\R}{\mathbb{R}}

\renewcommand{\P}{\mathbb{P}}

\newcommand{\bs}{\boldsymbol}

\usepackage{setspace}

\allowdisplaybreaks

\begin{document}

\title{Inference in a generalized Bradley-Terry model for paired comparisons with covariates and a growing number of subjects} 

\author{Ting Yan\thanks{Department of Statistics, Central China Normal University, Wuhan, 430079, China.
\texttt{Email:} tingyanty@mail.ccnu.edu.cn.}
\\
Central China Normal University
}
\date{}

\maketitle

\begin{abstract}
Motivated by the home-field advantage in sports, we propose a generalized Bradley--Terry model that incorporates covariate information
for paired comparisons. It has an $n$-dimensional merit parameter $\bs{\beta}$ and a $p$-dimensional regression coefficient
$\bs{\gamma}$ for covariates.
When the number of subjects $n$ approaches infinity and  the number of comparisons between any two subjects is fixed,
we show the consistency of the  maximum likelihood estimator (MLE) $(\widehat{\bs{\beta}}, \widehat{\bs{\gamma}})$ of $(\bs{\beta}, \bs{\gamma})$
in cases that $p$ is fixed or diverges with $n$.
Furthermore, when $p$ is fixed, we derive the asymptotic normal distribution of the MLE by characterizing its asymptotic representation.
The asymptotic distribution of $\widehat{\bs{\gamma}}$ is biased, while that of $\widehat{\bs{\beta}}$ is not.
This phenomenon can be attributed to the different convergence rates of $\widehat{\bs{\gamma}}$ and $\widehat{\bs{\beta}}$.
The consistency result is further extended to an Erd\H{o}s--R\'{e}nyi comparison graph.
Numerical studies and a real data analysis demonstrate our theoretical findings.

\vskip 5 pt \noindent
\textbf{Key words}: Asymptotic normality, Bradley-Terry model, Consistency, Covariate, Growing number of parameters. \\

{\noindent \bf Mathematics Subject Classification:} 60F05, 62J15, 62F12, 62E20.
\end{abstract}

\vskip 5pt

\section{Introduction}
When it was difficult for subjects to rank simultaneously based on the judgment of one person, they were arranged as follows:
repeatedly compared to each other in pairs.
Paired comparison data also arise in situations in which there are natural win-loss results between two subjects without the presence of a judge.
Subjects could be
teams, players, beverages, journals, and products.
One of the fundamental problems in paired comparisons is the production of a ranking for all subjects.
Because global rankings are not easily obtained for non round-robin tournaments,
developing a statistical model to estimate rankings is desirable.
The Bradley--Terry model \citep{bradley1952rank} is one of the most popular models for this purpose, dating back to at least 1929 \citep{zermelo1929berechnung}.
This is occasionally referred to as the Bradley--Terry--Luce model \citep{Luce-1959}.
It assigns one merit parameter $\beta_i$ to each subject and postulates that subject $i$ beats subject $j$ with a probability
$e^{\beta_i-\beta_j}/( 1 + e^{\beta_i-\beta_j})$, independent of other comparisons. The ranking of all the subjects is determined by
their maximum likelihood estimator (MLE).
Since the work of \cite{bradley1952rank}, it has had numerous applications, ranging from rankings of
classical sports teams \citep{masarotto2012the, Sire_2008baseball, Whelan-2020-Hockey}
and scientific journals \citep{stigler1994citation, Varin-2016-jrsa}
to the quality of product brands \citep{radlinski2007active},
such as two brands of wine of some type,
for multiclass classification \citep{hastie1998classification,clmenon2020multiclass}
and crowdsourcing \citep{chen2016overcoming}.

Among many paired comparison models, the Bradley--Terry model is unique,
satisfying the Bradley--Terry--Luce system \citep{colonius1980representation}.
\cite{hajek2014minimax} and \cite{shah2016estimation} demonstrated that the MLE in the Bradley--Terry model
is minimax-optimal for estimating the merit parameters in terms of the mean squared error.
A detailed investigation of the Bradley--Terry model, including
maximum likelihood estimation, hypothesis testing and
goodness-of-fit tests of the model can be found in Section 4 of \cite{David1988}.
To facilitate a wide range of applications, some generalized models have been proposed  \cite[e.g.][]{Luce-1959,rao1967ties,davidson1970on,huang2006generalized}.
Algorithms for solving MLEs in these models have been established \cite[e.g.][]{Ford1957, hunter2003mm, vojnovic2019convergence}.

As highlighted by \cite{agresti2012} (p. 455), most sports have
home-field advantage: A team is more likely to win when playing in its home city.
He introduced  a ``home-field advantage" model by assuming the logit of the probability of  home $i$ beating away $j$ is
the merit difference  $\beta_i-\beta_j$ plus an effect parameter $\gamma$, where $\gamma$ indicates a home-field advantage if $\gamma>0$.
The home team of the two evenly matched teams has the probability $\exp(\gamma)/(1 + \exp(\gamma))$ of winning. {\color{black}{Except the home-field advantage, there exist other covariates that have influences on the winning probabilities.
For example,  the duration of athletes' rest periods can influence competition results in sports competitions, where
athletes experiencing highly congested
schedules may be particularly exposed to air travel effects, circadian rhythm disruption, sleep deprivation, decline in physical capacity, and injury risk \citep{2020Basketball}.
Another example is that one can use a gambling company's odds to predict the outcome of a game \citep{song2026semiparametric},
where a team with low odds has a relatively stronger merit to win another team with large odds.
In addition, when a judge expresses a preference between two items, multiple covariates may affect the outcome, including
intrinsic covariates (e.g., the judge's professional background, judging experience, subjective predispositions, attention, and emotional state during assessment) and
extrinsic covariates (e.g., the evaluation environment, time constraints and stimulus presentation order).
}}
We extend the ``home-field advantage" model to a general form.

Let $\bm{z}_{ijk}$ be a deterministic $p$-dimensional vector denoting covariate information
associated with the $k$th comparison between subjects $i$ and $j$.
It is suitable to require $\bm{z}_{ijk}=-\bm{z}_{jik}$, because
if something is advantageous to $i$ then it is disadvantageous to $j$.
We incorporate the covariate information into the Bradley--Terry model by specifying the winning probability of $i$ against $j$ as
\begin{equation}
\label{model}
\P( i \mbox{~wins~} j |\bm{z}_{ijk}, \bs{\gamma}, \beta_i, \beta_j ) = \frac{ \exp( \beta_i - \beta_j + \bm{z}_{ijk}^\top \bs{\gamma}) }{
1 + \exp( \beta_i - \beta_j + \bm{z}_{ijk}^\top \bs{\gamma}) },
\end{equation}
where $\bs{\gamma}$ is a $p$-dimensional regression coefficient of the covariates and $\beta_i$ is the merit parameter of $i$.
Under the restriction $\bm{z}_{ijk}=-\bm{z}_{jik}$, the probability above is well defined.
We call it  the covariate-adjusted Bradley--Terry model (abbreviated as ``CBTM") hereafter.

The covariate $\bm{z}_{ijk}$ can be formalized according to the situations of the teams or the attributes of the subjects.
If $\bm{x}_{ik,j}$ and $\bm{x}_{jk,i}$ denote $p$-dimensional attributes of subjects $i$ and $j$ in the $k$th comparison of $i$ against $j$, respectively, they can be used to construct vector $\bm{z}_{ijk}=\mathbf{g}(\bm{x}_{ik,j}, \bm{x}_{jk,i})$ for an asymmetric vector function, where
$\mathbf{g}(\bm{x}_{ik,j}, \bm{x}_{jk,i})=-\mathbf{g}(\bm{x}_{jk,i},\bm{x}_{ik,j})$. For instance, if we let $\mathbf{g}(\bm{x}_{ik,j}, \bm{x}_{jk,i})$
be equal to $\bm{x}_{ik,j} - \bm{x}_{jk,i}$, then we can measure the dissimilarity between the two subjects.
As an example, if the game is played in the city of team $i$ or subject $i$ is listed first,
then we let $x_{ik,j}=1$ and $X_{jk,i}=0$ ($p=1$), such that $z_{ijk}=1$ and $z_{jik}=-1$.
In this case, the CBTM reduces to the home-field advantage model in \cite{agresti2012}.

In several paired comparison situations,
the number of subjects $n$ is typically large, whereas the number of comparisons $m_{ij}$ for any pair $(i,j)$ is relatively small
\citep[e.g.][]{simons-yao1999}.
Specifically, all $m_{ij}$ are bounded by a fixed constant (e.g., each pair of teams in the NBA plays at most four games in a regular season).
Theoretical analysis of the Bradley--Terry model with a diverse number of subjects has received wide attention in recent years.
We have elaborated on these after stating our main results.
However, the existing high-dimensional paired comparison literature has little involvement in
additional information (e.g., covariate), but for win-loss outcomes.
As shown in \cite{agresti2012}, this can significantly influence the ranking.

To determine how covariates influence the estimation of merit parameters in the high dimension,
we drew  Figure \ref{figure-intu} to evaluate the $\ell_\infty$-error $\|\widehat{\bs{\beta}} - \bs{\beta}\|_\infty$ via a simulation study here,
where  $\widehat{\bs{\beta}}$ was fitted using the maximum likelihood estimation in the original BTM (shorthand of Bradley--Terry model) and the CBTM with the correct model specification for comparison.
The figure on the left shows   
that the  error in the BTM increases rapidly with $\gamma$ whereas the error in the CBTM changes only slightly.
The former became increasingly larger than the latter for $\gamma$.
From the right figure, we can see the following:
the error in the BTM is larger than that in the CBTM when $n$ increases and $\gamma$ is fixed.
Even when $n$ increased to a sufficiently large value, the error in the BTM did not decrease.
This indicates that the CBTM has a significant improvement over the BTM when there is covariate effect.
This partly motivates the present study.

\begin{figure}[h]
  \centering
  \subfigure{
  \centering
    \includegraphics[width=0.45\linewidth, height=2.0in]{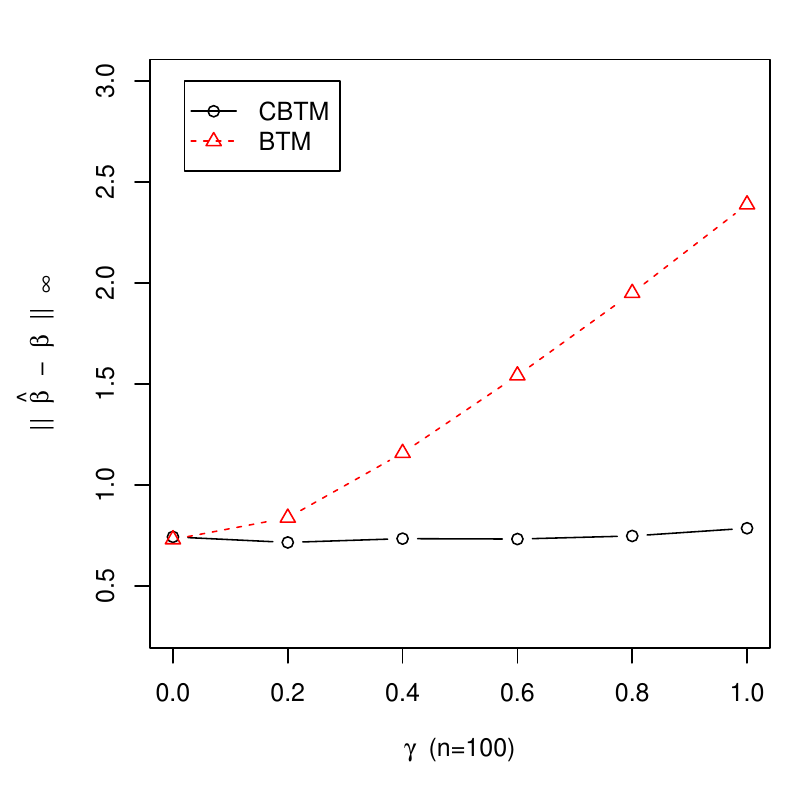}}
  \subfigure{
    \centering
    \includegraphics[width=0.45\linewidth, height=2.0in]{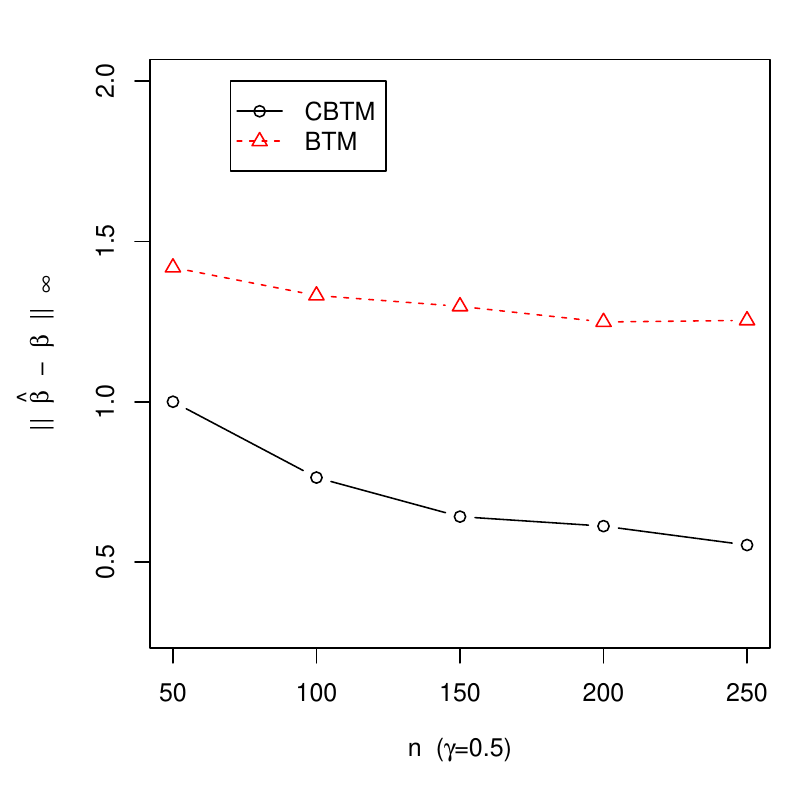}}
   \caption{ The plots of the average values of $\|\widehat{\bs{\beta}} - \bs{\beta} \|_\infty$ changing with $\gamma$ (when $n$ is fixed) in the left
   and changing with $n$ (when $\gamma$ is fixed) in the right.
   The red color indicates the error with fitted values in the BTM without covariates while
   the black color indicates the error in the CBTM. All $\beta_i$s were independently generated from the uniform distribution $U(0,1)$. Each pair had only one comparison.
We set $z_{ijk}=1$ when $i<j$ and $z_{ijk}=-1$ when $i>j$. 
The win-loss outcomes were generated according to the CBTM.
The average value of $\|\widehat{\bs{\beta}} - \bs{\beta}\|_\infty$ was recorded out of $100$ repetitions.
}
\label{figure-intu}
\end{figure}

The contributions of this study are as follows.
\begin{itemize}
\item
When $n$ goes to infinity and all $m_{ij}(>0)$ are fixed, we establish
the uniform consistency of the MLE {\color{black}{when the dimension of covariates is fixed or increases with $n$}}.
{\color{black}{
A key idea for the proof is that we use a two-stage method
that first derives an error bound for the MLE $\widehat{\bs{\gamma}}$ of the regression coefficient
$\bm{\gamma}$ of the covariates and then obtains
the $\ell_\infty$-error between an estimator $\widehat{\bs{\beta}}_{\gamma}$ and $\bs{\beta}$ for a given $\bs{\gamma}$ via the convergence rate of a Newton's iterative sequence.
The proof in the first stage is delicate, where we
derive the lower bound of
the difference of the expected log-likelihood function between the true parameter and its neighboring point and  a uniform lower bound of the averaged centered log-likelihood function between the true parameter and its its neighboring point, and connect it to the error bound of the true parameter.
}}

\item
We derive the asymptotic normal distribution of the MLE by characterizing its asymptotic representation, \textcolor{black}{when the dimension of covariates is fixed}.
This is proved by applying Taylor's expansions to a series of functions constructed from likelihood equations and showing
remainder terms in the expansions are  asymptotically neglect.
The asymptotic distribution of the MLE $\widehat{\bs{\gamma}}$ contains a bias term
while there is no bias for $\widehat{\bs{\beta}}$.
We propose a bias-corrected version of $\widehat{\bs{\gamma}}$ for making valid inference.

\item
We further extend the consistency result to  an Erd\H{o}s-R\'{e}nyi random graph
with a diverging number of $p_n$, where the sampling probability is allowed to be close to the Erd\H{o}s-R\'{e}nyi lower bound \citep{erd6s1960evolution}.
We also extend the results to a fixed sparse comparison graph in \cite{yan2012sparse}.

\end{itemize}

Simulation studies and a real data analysis are conducted to illustrate the theoretical results.

\subsection{Related work}
\label{section:related}

Studies on the Bradley--Terry model in high-dimensional settings have recently attracted significant interest.
In a pioneering study, \cite{simons-yao1999} 
proved the uniform consistency and asymptotic normal distribution of the MLE 
when the number of subjects approaches infinity and each pair has a fixed number of comparisons.
To relax the dense comparison assumption, \cite{yan2012sparse} extended their results to a fixed sparse comparison graph by controlling the length from one
subject to another subject. 
\cite{Han-chen2020} further extended \citeauthor{simons-yao1999}'s results to an Erd\"{o}s--R\'{e}nyi comparison graph under a weak sparsity condition on $q_n$,
where $q_n$ is the probability that any two subjects will be compared.
\cite{chen2019} established the $\ell_\infty$-error bounds for the spectral estimator 
and regularized the MLE, which leads to
sample complexity for the top-$K$ rankings.
\cite{chen2021optimal} further studied the $\ell_\infty$-error of the MLE and obtained the minimax rate for top-$k$ ranking.
However, covariate information was not considered in these studies, which is the focus of this study.

While revising this paper\footnote{An original version was finished in 2020.}, a new related work appears.
\cite{fan2024uncertainty} extend the Bradley--Terry model to incorporate the covariate information,
where the covariate term is $(X_i - X_j)^\top \bs{\gamma}$ and
$X_i$ denotes the individual-level attribute $X_i$ of subject $i$. In contrast,
the covariate term in model \eqref{model} is $Z_{ijk}^\top \bs{\gamma}$, which contains the special case $(X_i - X_j)^\top \bs{\gamma}$.
It is clear that  \cite{fan2024uncertainty} characterize only the individual level covariate information and do not address
such covariates associated with each paired comparision (e.g.,  home-field advantage).
In addition, our proof strategy is different from theirs, where \cite{fan2024uncertainty} analyze the consistency of the MLE by using a constrained maximum
likelihood technique with a projected gradient descent algorithm and derive asymptotic distributions of the MLE by approximating the MLE via the minimizer of the quadratic approximation of the likelihood function.
In this study, we use a two-stage technique that alternatively obtains the $\ell_\infty$-error between an estimator $\widehat{\bs{\beta}}_{\gamma}$ and $\bs{\beta}$ and
 the $\ell_\infty$-error between an estimator $\widehat{\bs{\gamma}}_{\beta}$ and $\bs{\gamma}$, to show the consistency of the MLE and characterize asymptotic representations to obtain its asymptotic distributions.

Note that the CBTM can be recast into a logistic regression model.
The ``large $N$, diverging $p_N$" framework in generalized linear models (GLMs) has been explored, where
$N$ is the sample size, and $p_N$ is the dimension of the parameter space.
\cite{portnoy1988} showed the asymptotic normality of the MLE in exponential family of distributions on independent and identically distributed samples when $p_N^2=o(N)$.
\cite{HE2000120} built the asymptotic normality of $M$-estimators when $p_N^2\log p_N = o(N)$.
\cite{wang2011} established the consistency of the generalized estimating equations estimator when
$p^2_N = o(N)$ and its asymptotic normality when  $p^3_N = o(N)$.
In our asymptotic framework for the CBTM, $p_N^2/N \to 1/2$, not $0$, where $p_N = p+n$ and $N=n(n+1)/2$ if each pair has only one comparison.
Therefore, these asymptotic results are not applicable in this case.

A relevant work to GLMs is \cite{liang2012}, who study the asymptotic regime $p_N=o(N)$ in a logistic regression model.
Let $\lambda_{\min}(S_N)$ and $\lambda_{\max}(S_N)$ denote the minimum and maximum eigenvalues of $S_N$,
where $S_N=\sum_{i=1}^N x_i x_i^\top$ and $x_i$ is the $p_N$-dimensional covariate vector of individual $i$.
Assuming that
$c_1 N \le \lambda_{\min}(S_N) \le \lambda_{\max} (S_N) \le c_2 N$ for two constants $c_1$ and $c_2$, they show the asymptotical normality of  the MLE
by extending the proof strategy in \cite{yin2006asymptotic} for GLMs with fixed dimensions to an
increasing dimension.
In CBTM, 
the first $n$ diagonal entries of $S_N$ are of the order of $n$, because of the special structure
of the design matrix for the merit parameters $\bs{\beta}$, whereas the last $p$ diagonal entries of $S_N$ are on the order of $n^2$.
Because of the different orders of the diagonal elements of $S_N$, the ratio $\lambda_{\max} (S_N)/\lambda_{\min}(S_N)$ is not constant. In a broad simulation study, we found the following:
$\lambda_{\max} (S_N)/\lambda_{\min}(S_N)$ is of the order of $O(N)$, far from the assumption that
$ \lambda_{\max} (S_N)/\lambda_{\min}(S_N) \le c_2/c_1$, Therefore, the conditions in \cite{liang2012} cannot be applied to CBTM.
Interestingly, a recent study reported the following:
\cite{ZHOU2021107154} extended \citeauthor{yin2006asymptotic}'s proof to GLMs
with a diverging number of covariates, which requires, except for the same conditions for $S_N$ as in \cite{liang2012},
 the condition $p_N^2/N\to 0$ to guarantee asymptotic normality rather than the weaker condition $p_N/N\to 0$. 
In addition,
the asymptotic distribution of the MLE in the aforementioned literature is not biased  \cite[e.g.][]{haberman1977,portnoy1988,wang2011,liang2012,ZHOU2021107154}.
In sharp contrast to these studies, the asymptotic distribution of MLE $\widehat{\bs{\gamma}}$
has a bias term, whereas that of the MLE $\widehat{\bs{\beta}}$ does not.
This phenomenon is referred to as the incidental parameter problem in
econometric literature \cite[e.g.,][]{Graham2017} 
which is caused by different convergence rates  of $\widehat{\bs{\gamma}}$ and $\widehat{\bs{\beta}}$.

In the network setting, the degree heterogeneity and the homophily have been modelled in a similar logistic regression form
\cite[e.g.][]{Graham2017,Yan-Jiang-Fienberg-Leng2018}. However, their focus are network features, which is different.
In addition, the case with the increasing dimension of covariates is not studied in their works.
Model \eqref{model} can also be represented as a log-linear model.
Although the conditions for
the existence of an MLE  have been established \citep{ifenberg2012-Rinaldo},
asymptotic theories remain lacking in high-dimensional cases \cite[e.g.][]{Fienberg:Rinaldo:2007, ifenberg2012-Rinaldo}.

The remainder of this paper is organized as follows. In Section \ref{section:model}, we present the maximum likelihood estimation.
In Section \ref{section:asymptotic}, we present the consistency and asymptotic normality of the MLE.
In Section \ref{section:extension}, we extend the consistency result to an Erd\H{o}s--R\'{e}nyi comparison graph
with a diverging number of covariates.
In Section \ref{section:simulation}, we perform a simulation and provide a real data analysis.
We provide a summary and further discussion in section \ref{section:sd}.
All proofs are presented on the supplementary material.

\section{Maximum likelihood estimation}
\label{section:model}

Consider a set of $n+1$ subjects  labelled by ``$0, \ldots, n$".
Let $m_{ij}$ be the number of comparisons between $i$ and $j$ and
$A_{ijk}$ be the random outcome in the $k$th comparison, $k=1, \ldots, m_{ij}$, where
$A_{ijk}$ ($\in \{0, 1\}$) is an indictor variable denoting whether $i$ beats $j$ in the $k$th comparison. {\color{black}{The lowercase letter $a_{ijk}$ denotes its realized value.}}
That is, if $i$ wins $j$, then $a_{ijk}=1$; otherwise, $a_{ijk}=0$.
The win-loss results are recorded in a matrix:
 $A=(A_{ij})_{n\times n}$, where
$A_{ij}$ is the number of times that $i$ beats $j$ and the diagonal elements $A_{ii}$ are set to zero by
default, i.e., $A_{ii}=0$.
Let $d_i= \sum_{j \neq i} A_{ij}$ be the total number of wins for subject $i$ and $\bs{\beta}=(\beta_1, \ldots, \beta_n)^\top$.
We make the following assumptions.

{\color{black}{
\begin{assumption}
\label{assump-m-star}
All \(m_{ij}\) are uniformly bounded above by a fixed positive constant $m_*$.
\end{assumption}

\begin{assumption}
\label{assump-a-indepen}
All the comparison outcomes \(A_{ijk}\) over $0\le i<j \le n, j=1, \ldots, m_{ij}$, are mutually independent.
\end{assumption}

When the number of subjects is large, Assumption \ref{assump-m-star} usually holds.
For instance, two teams play a maximum of four times in the NBA regular season.
Assumption \ref{assump-a-indepen} requires the independence of paired comparisons,
which is widely used in literature \cite[e.g.][]{simons-yao1999,Han-chen2020,fan2024uncertainty}.
}}

Because adding the same constant to all $\beta_i$ results in the invariance of probability \eqref{model},
we set $\beta_0=0$ for model identification, as in \cite{simons-yao1999}. Other restrictions are possible; for example, $\sum_i \beta_i=0$.
{\color{black}{
Except for this restriction, some condition is needed for the covariate matrix $Z$ for guaranteeing
the model identifiability.
Let $N=\sum_{i<j}m_{ij}$ be the total number of comparisons
and $\mathbf{e}_i\in\mathbb{R}^{n}$ denote the $i$th standard basis vector with its $i$th element being $1$ and all other elements $0$.
Let $\mathbf{x}_{ij}^\top = ( \mathbf{x}_{ij, 1}^\top, \ldots, \mathbf{x}_{ij, m_{ij} }^\top)^\top$
and
\begin{equation}
\label{definition-X}
X = (\mathbf{x}_{01}^\top, \ldots, \mathbf{x}_{0n}^\top,
\mathbf{x}_{12}^\top, \ldots, \mathbf{x}_{1n}^\top, \ldots,
\mathbf{x}_{n-1,n}^\top)^\top \in \mathbb{R}^{N\times n},
\end{equation}
where
$\mathbf{x}_{0j,k}=-\mathbf{e}_j$ for $j=1,\ldots,n$ and $k=1,\ldots, m_{0j}$, and
$\mathbf{x}_{kj, t}=\mathbf{e}_k - \mathbf{e}_j$ for $1\le k \le n-1$, $j=k+1, \ldots, n$
and $t=1,\ldots,m_{kj}$.
It is clear that $X$ is  the design matrix for the parameter vector $(\beta_1, \ldots, \beta_{n})$.
Recall that $\mathbf{z}_{ij, k}\in \mathbb{R}^p$ denotes the $p$-dimensional covariate vector
associated with $k$-th comparison when subject $i$ is compared with subject $j$.
Correspondingly, define $\mathbf{z}_{ij}^\top = ( \mathbf{z}_{ij, 1}^\top, \ldots, \mathbf{z}_{ij, m_{ij} }^\top)^\top$ and
\begin{equation}
\label{definition-Z}
Z = (\mathbf{z}_{01}^\top, \ldots, \mathbf{z}_{0n}^\top,
\mathbf{z}_{12}^\top, \ldots, \mathbf{z}_{1n}^\top, \ldots,
\mathbf{z}_{n-1,n}^\top)^\top \in \mathbb{R}^{N\times p}.
\end{equation}
That is, $Z$ is the design matrix for the $p$-dimensional regression parameter $\bm{\gamma}$.
Let $(X,Z)$ be the merged matrix of $X$ and $Z$ by column.
The following proposition gives the identifiability of model, whose proof is in the supplementary material.

\begin{proposition}
\label{prop-model}
If $\mathrm{rank}(X,Z) = n + p$, then the covariate-adjusted Bradley--Terry model is identical, subject to the restriction
$\beta_0=0$.
\end{proposition}

We now provide some intuitive explanations on the condition $\mathrm{rank}(X, Z)=n+p$.
The matrix  $X (X^\top X)^{-1}X^\top \in\mathbb{R}^{N\times N}$
 is the projection matrix that project a vector onto the linear subspace $\mathcal{C}({X^\top})$ spanned by the column vectors of $X^\top$.
 Let $I_N$ denote the $N\times N$ identity matrix.
 Then, we have
\begin{equation}
\label{Z-decomposition}
Z= X (X^\top X)^{-1}X^\top Z + (I_N - X (X^\top X)^{-1}X^\top)Z,
\end{equation}
where the first term of the above right hand side represents
the projection of $Z\in \mathbb{R}^{N\times p}$
onto the linear subspace $\mathcal{C}({X^\top})$ and the second term is
the residual vector. By the fact $(I_N - X (X^\top X)^{-1}X^\top) X =0$, we have
\begin{align}\label{eq:identification}
Z^\top P_X \mathbf{h} = Z^\top P_X (X \bm{\beta} + Z \bm{\gamma} )
= (Z^\top P_X Z) \bm{\gamma},
\end{align}
where $P_X=I_N - X (X^\top X)^{-1}X^\top$ and $\mathbf{h}$ is defined in \eqref{definition-h}
denoting the win-loss probability vector.
If $\mathrm{rank}(P_X Z) < p$, then the regression coefficient $\bm{\gamma}$ of the covariate
is not identical.
A simple example is that $\mathbf{z}_{ijk}=\mathbf{x}_i - \mathbf{x}_j$ for all $i,j,k$, where
$\mathbf{x}_i$ denotes the attribute vector of subject $i$. In this case, $\mathrm{rank}(P_X Z)=0$.
 On the other hand, if $\mathrm{rank}(P_X Z) = p$, by \eqref{eq:identification}, then
\[
\bm{\gamma} = (Z^\top P_X Z)^{-1}Z^\top P_X \mathbf{h}.
\]
That is, $\bm{\gamma}$ is identifiable.
On the other hand, since $\mathrm{rank}(X^\top X) = n$, if $\mathrm{rank}(P_X Z) = p$, then all the parameters
are identical.
}}

Under Assumption \ref{assump-a-indepen},
 the log-likelihood function in model \eqref{model},  is
\renewcommand{\arraystretch}{1.3}
\[
\begin{array}{rcl}
\ell(\bs{\beta}, \bs{\gamma})
& = & \sum\limits_{0 \le i < j \le n} \sum\limits_{k=1}^{m_{ij}} \{ A_{ijk}(\beta_i - \beta_j + \mathbf{z}_{ijk}^\top \bs{\gamma} )
- \log ( 1 + e^{\beta_i - \beta_j + \mathbf{z}_{ijk}^\top \bs{\gamma} } ) \} \\
& = &
\sum\limits_{i}  \beta_i D_i + \sum\limits_{ i<j } \sum\limits_{k} A_{ijk} (\mathbf{z}_{ijk}^\top \bs{\gamma})
 - \sum\limits_{i<j} \sum\limits_{k} \log ( 1 + \exp( \beta_i - \beta_j + \mathbf{z}_{ijk}^\top \bs{\gamma})).
\end{array}
\]
Write $\mu_{ijk}(\bs{\beta}, \bs{\gamma})$ as the expectation of $a_{ijk}$, where
$\mu_{ijk}(\bs{\beta}, \bs{\gamma})$ is equal to the probability of $i$ winning $j$ in the $k$th comparison given in \eqref{model}.
The maximum likelihood equation is as follows:
\renewcommand{\arraystretch}{1.3}
\begin{equation}\label{eq:moment}
\begin{array}{rcl}
d_i & = & \sum\limits_{j=0,j\neq i}^n \sum\limits_{k=1}^{m_{ij}}\mu_{ijk}(\bs{\beta}, \bs{\gamma}),~~ i=1, \ldots, n, \\
\sum\limits_{0\le i<j \le n}  \sum\limits_{k=1}^{m_{ij}} A_{ijk}\mathbf{z}_{ijk} & = & \sum\limits_{0\le i<j \le n}  \sum\limits_{k=1}^{m_{ij}} \mathbf{z}_{ijk} \mu_{ijk}(\bs{\beta}, \bs{\gamma}).
\end{array}
\end{equation}
It should be noted that the above equations do not contain $d_0$. This is because $\sum_{i=0}^n d_i = \sum_{i<j} m_{ij}$,
The MLE of the parameter $(\bs{\beta},\bs{\gamma})$, denoted as $(\widehat{\bs{\beta}}, \widehat{\bs{\gamma}})$, is the solution to the above
equations,  where $\hat{\beta}_0=0$.

Let $\mathcal{K}$ be the convex hull of
set
\[
\{(d_1, \ldots, d_{n}, \sum_{i<j} \sum_k \mathbf{z}_{ijk}^\top A_{ijk})^\top: A_{ijk}\in \{0,1\}, 0\le i<j \le n, k=1, \ldots, m_{ij} \}.
\]
If $\mathrm{rank}(X,Z) = n + p$, then, according to Theorem 1.13 in
\cite{Brown-1986} (p. 19), the  normalizing function $\sum_{i<j}\sum_k \log ( 1 + \exp (\beta_i -\beta_j + \mathbf{z}_{ijk}^\top \bs{\gamma}))$ is strictly
convex.
As this  normalizing function is steep and strictly
convex, by the properties of exponential family of distributions [e.g., Theorem 5.5 in \cite{Brown-1986} (p. 148)], we have the following result.

\begin{proposition}
\label{pro-existence}
{\color{black}If $\mathrm{rank}(X, Z) = n + p$,}
the MLE $(\widehat{\bs{\beta}}, \widehat{\bs{\gamma}})$ exists if and only if
\[
(d_1, \ldots, d_{n}, \sum_{i<j} \sum_k \mathbf{z}_{ijk}^\top A_{ijk})^\top
\]
lies in the interior of $\mathcal{K}$.
{\color{black}Further, if the MLE exists, it is unique.}
\end{proposition}

If the vector $(d_1, \ldots, d_{n})$ contains zero elements (corresponding to subjects without wins),
or values being equal to the total number of comparisons of some subjects (corresponding to subjects without losses), the MLE did not exist.
If we do not consider covariate information, then the condition in Proposition \ref{pro-existence} can be explained in terms of graph language.
The win-loss matrix $A$ can be represented in a directed graph $\mathcal{G}_n$ with nodes denoting subjects
and directed edges denote the number of wins for one subject against another.
The necessary and sufficient condition for the existence of $\widehat{\bs{\beta}}$
is that the directed graph $\mathcal{G}_n$ is strongly connected. That is,  for every partition of subjects into two nonempty sets, a
subject in the second set beats the subject in the first set at least once [\cite{Ford1957}].

{\color{black}{
Note that there is a restriction on regression coefficient in semiparametric paired comparison models of \cite{song2026semiparametric}.
\cite{song2026semiparametric} treat the covariates as random vectors while they are
determined in this manuscript.
\cite{song2026semiparametric} employ a special regressor method to obtain
explicit expressions for all unknown parameters, where one element of the covariate vector
is specially chosen to transform the original paired comparison data into
the working data.
Roughly speaking, the regression coefficient of that special regressor is not needed to estimate in their method, which is different from the maximum likelihood estimation here.
}}

We discuss computational issues.
For moderate $n$,
we can simply use the package ``glm" in the R language to solve the MLE, which
implements an iteratively reweighted least squares algorithm.
{\color{black}{
The computational complexity of an iteratively reweighted least squares algorithm is $O( K(N(n+p_n)^2 + (n+p_n)^3) )$, where $K$ is the number of iterations
until the convergence condition is met.
The memory requirements for this algorithm are dominated by the storage of the $N\times (n+p_n)$ design matrix of the parameter vector $(\bm{\beta},\bm{\gamma})$,
Its storage complexity is $O(Np_n)$ because each row of the $N\times n$ design matrix of the parameter vector $\bm{\beta}$ have only at most two nonzero elements.
Its storage complexity is $O(Np_n)$ because each row of the $N\times n$ design matrix of the parameter vector $\bm{\beta}$ have only at most two nonzero elements.
Assuming that each pair is compared only once and $p_n=2$, the memories for the design matrix take $4Mb$, $12.5Mb$, $49.5Mb$  when $n=100$, $500$, $1000$, respectively.
When $n=200$, the memories for the design matrix take $4Mb$, $12.5Mb$, $49.5Mb$  when $p_n=30$, $50$, $100$, respectively.
In this case, the corresponding running times are $2.89$, $3.35$, $4.76$ seconds
in a computer with CPU Intel i9-10900 2.8GHz and  64GB RAM when applying the ``glm" package to solve the MLE on a data set.
}}
For relatively large $n$, it might not have large enough memory to store the design matrix for $\bs{\beta}$ required by the ``glm."
In this case,
we recommend using a two-step iterative algorithm by alternating between solving the first equation in  \eqref{eq:moment} using the fixed-point method in \cite{Ford1957}
and solving the second equation in \eqref{eq:moment} using the Newton-Raphson method.  
{\color{black}{
In the first stage, the iteration is
\[
\beta_i^{(k+1)} = \log d_i  - \sum_{j\neq i} \sum_k \frac{ m_{ijk} e^{ - \beta_j^{(k)} + \bm{z}_{ijk}^\top \bm{\gamma} }}{ 1 + e^{ \beta_i^{(k)} - \beta_j^{(k)} + \bm{z}_{ijk}^\top \bm{\gamma} }},
i=1, \ldots, n.
\]
Its computational complexity is $O(n^2p_n)$ in each iteration.
Recall that we define the function
\[
Q_\beta(\bm{\gamma}) =  \sum\limits_{0\le i<j \le n}  \sum\limits_{k=1}^{m_{ij}} \mathbf{z}_{ijk} \mu_{ijk}(\bs{\beta}, \bs{\gamma}) - \sum\limits_{0\le i<j \le n}  \sum\limits_{k=1}^{m_{ij}} a_{ijk}\mathbf{z}_{ijk}.
\]
The Newton-Raphson iteration is
\[
\bm{\gamma}^{(k+1)} = \bm{\gamma}^{(k)} - [Q_{\beta}(\bm{\gamma}) ]^{-1}Q_{\beta}(\bm{\gamma}).
\]
Its computational complexity is $O(p_n^3N)$ in each iteration. The total computational complexity
is $O( K ( n^2p_n + p_n^3N) )$, where $K$ is the number of iterations.
}}

\section{Theoretical properties}
\label{section:asymptotic}
In this section, we present the consistency and asymptotic normality of the MLE.
First, we introduce certain notation. For a subset $C\subset \R^n$, let $C^0$ and $\overline{C}$ denote the interior and closure of $C$, respectively.
For a vector $\mathbf{x}=(x_1, \ldots, x_n)^\top\in \R^n$, we denote
$\|\mathbf{x}\|_\infty = \max_{1\le i\le n} |x_i|$ and $\|\mathbf{x}\|_1=\sum_i |x_i|$ by the $\ell_\infty$- and $\ell_1$-norms of $\mathbf{x}$, respectively.
Let $B(\mathbf{x}, \epsilon)=\{\mathbf{y}: \| \mathbf{x}-\mathbf{y}\|_\infty \le \epsilon\}$ be the $\epsilon$-neighborhood of $\mathbf{x}$.
For an $n\times n$ matrix $J=(J_{ij})$, let $\|J\|_\infty$ denote the matrix norm induced by the $\ell_\infty$-norm on the vectors in $\R^n$; that is,
\[
\|J\|_\infty = \max_{\mathbf{x}\neq 0} \frac{ \|J\mathbf{x}\|_\infty }{\|\mathbf{x}\|_\infty}
=\max_{1\le i\le n}\sum_{j=1}^n |J_{ij}|,
\]
where $\|J\|$ denotes a general matrix norm.
Define the maximum absolute entry-wise norm: $\|J\|_{\max}=\max_{i,j}|J_{ij}|$.
We use the superscript ``*" to denote the true parameter under which the data are generated.
When there is no ambiguity, we omit the superscript ``*".
The notation $\sum_{i<j}$  is a shorthand for $\sum_{i=0}^n \sum_{j=i+1}^n$.
Define
\begin{equation}\label{definition-pi}
\mu(x):=\frac{e^x}{1+e^x}, ~~\pi_{ijk}:= \beta_i - \beta_j + \mathbf{z}_{ijk}^\top \bs{\gamma}, ~~\pi_{ijk}^*:= \beta_i^* - \beta_j^* + \mathbf{z}_{ijk}^\top \bs{\gamma}^*.
\end{equation}
Note that $\mu(\pi_{ijk}^*)$ as the expectation of $a_{ijk}$.
We will use the notations $\mu(\pi_{ijk})$ and $\mu_{ijk}(\bs{\beta}, \bs{\gamma})$ interchangeably.
$c, c_0, c_1$, $C, C_0, C_1$, $\ldots$, refer to universal constants, whose
specific values may vary from place to place.
When $p$ diverges with $n$, we write $p_n$.

{\color{black}{
Let $\bm{\theta}=(\bm{\beta}^\top, \bm{\gamma}^\top)^\top$ and
\[
\Theta = \{ \bm{\theta}: \|\bm{\beta} - \bm{\beta}^*\|_\infty \le 1, \|\bm{\gamma}-\bm{\gamma}^*\|_\infty \le 1 \}.
\]
When emphasizing $\pi_{ijk}$ on its arguments, we write $\pi_{ijk}(\bm{\theta})$.
We make the following assumption on parameters.

\begin{assumption}
\label{assump-parameter}
We assume that there exist two numbers $\kappa_{n}$ and $b_n$ that may diverge with $n$ such that
\begin{equation}
\label{definition-kappa_n}
\sup_{i,j,k}\|\bm{z}_{ijk}\|_\infty \le \kappa_n \le \log n,
\end{equation}
and
\begin{equation}
\label{definition-b_n}
\sup_{\bm{\theta}\in \Theta} \max_{i,j,k} \frac{ ( 1 + e^{\pi_{ijk}(\bm{\theta})})^2 }{ e^{\pi_{ijk}(\bm{\theta}) } } \le b_n.
\end{equation}
\end{assumption}
}}

It is clear that $b_n\ge 4$ and
\[
 \inf_{\bm{\theta}\in \Theta} \min_{i,j,k} \frac{ e^{\pi_{ijk}(\bm{\theta}) }  }{( 1 + e^{\pi_{ijk}(\bm{\theta})})^2 } \ge 1/b_n.
\]
When all covariates $\mathbf{z}_{ijk}$ are bounded by a constant, \eqref{definition-kappa_n} automatically holds.
This condition is presented in \cite{Graham2017}, \cite{Dzemski2019} and \cite{Yan-Jiang-Fienberg-Leng2018}.
Further, we make the following assumption on covariates.

\begin{assumption}
\label{assump-residual}
Let $\tilde{Z}=(I_N - X (X^\top X)^{-1}X^\top)Z$. We assume that there exists a positive fixed constant $c_0$ such that
Therefore,
\[
\lambda_{\min}\left(\frac{ \tilde{Z}^\top\tilde{Z} }{ N }\right)
\ge c_0>0.
\]
\end{assumption}

{\color{black}{
In the above assumption, $\tilde{Z}$ could be viewed as
the residual vector after projecting $Z$ into the column space $\mathbb{C}(X^\top)$ as mentioned before, where $X\in \R^{N\times p}$ is the design matrix for the $\beta$ parameter.
Recall that $Z= X (X^\top X)^{-1}X^\top Z + \tilde{Z}$ in \eqref{Z-decomposition},
where the first and second terms  are orthogonal.
An intuitive explanation for Assumption \ref{assump-residual} is that the amount of information
of the covariate matrix $Z$ projected into $\mathbb{C}(X^\top)^+$ is linearly
proportional to the total amount of information, where $\mathbb{C}(X^\top)^+$ denotes the orthogonal complementary space of $\mathbb{C}(X^\top)^+$.
A sufficient condition to guarantee Assumption \ref{assump-residual}
is that $\bm{z}_{ijt}$'s are independently generated from
 a $p$-dimensional non-degenerated multivariate symmetric distribution, whose proof is omitted.
 In practice, since $\tilde{Z}$ is observed,
the eigenvalues of $\tilde{Z}^\top\tilde{Z}$ can be computed to judge
whether Assumption \ref{assump-residual} holds.
}}

To explore the asymptotic properties of the MLE, we introduce a system of score functions based on the maximum likelihood equations:
\begin{equation}\label{eqn:def:F}
\begin{array}{rcl}
H_i(\bs{\beta}, \bs{\gamma}) & = & \sum\nolimits_{ j\neq i}\sum\nolimits_{k} \mu_{ijk}(\bs{\beta}, \bs{\gamma}) - d_i,~~i= 0, \ldots, n, \\
H(\bs{\beta}, \bs{\gamma}) & = & (H_1(\bs{\beta}, \bs{\gamma}), \ldots, H_{n}(\bs{\beta}, \bs{\gamma}))^\top.
\end{array}
\end{equation}
Furthermore, we define $H_{\gamma,i}( \bs{\beta} )$ as the value of $H_i(\bs{\beta}, \bs{\gamma})$ for an arbitrarily fixed $\bs{\gamma}$, and
\[
H_\gamma(\bs{\beta})=(H_{\gamma,1}(\bs{\beta}), \ldots, H_{\gamma,n}(\bs{\beta}))^\top.
\]
Let $\widehat{\bs{\beta}}_\gamma$ be the solution to $H_\gamma(\bs{\beta})=0$.
Correspondingly, we define two additional score functions:
\begin{eqnarray}
\label{definition-Q}
Q(\bs{\beta}, \bs{\gamma}) & = & \sum\nolimits_{i<j} \sum\nolimits_k \mathbf{z}_{ijk} \{   \mu_{ijk}(\bs{\beta}, \bs{\gamma}) - a_{ijk} \}, \\
\label{definition-Qc}
Q_c(\bs{\gamma}) & = & \sum\nolimits_{i<j} \sum\nolimits_k \mathbf{z}_{ijk} \{ \mu(\widehat{\beta}_{\gamma,i} - \widehat{\beta}_{\gamma,j} +  \mathbf{z}_{ijk}^\top \bs{\gamma}) - A_{ijk} \}.
\end{eqnarray}
$Q_c(\bs{\gamma})$ can be viewed as a concentrated or profiled function of
$Q(\bs{\beta}, \bs{\gamma})$, where the merit parameter $\bs{\beta}$ was profiled.
Clearly, if $(\widehat{\bs{\beta}}, \widehat{\bs{\gamma}})$ exist, then
\begin{equation*}\label{equation:FQ}
H(\widehat{\bs{\beta}}, \widehat{\bs{\gamma}})=0,~~H(\widehat{\bs{\beta}}_\gamma, \bs{\gamma})=H_\gamma(\widehat{\bs{\beta}}_\gamma)=0,~~Q(\widehat{\bs{\beta}}, \widehat{\bs{\gamma}})=0,
~~Q_c(\widehat{\bs{\gamma}})=0.
\end{equation*}

\subsection{Consistency}
\label{subsec-consis}

Note that model \eqref{model} contains two sets of parameters: a merit vector parameter $\bs{\beta}$ with a growing dimension, and
regression coefficient $\bs{\gamma}$ of covariates with fixed dimensions.
In the absence of covariates, \cite{simons-yao1999} proved the consistency of the MLE through two key steps that first bound the probability that
the strong connection condition in the win-loss comparison graphs failed.
 Then, we find a set of common neighborhoods with ratios close to the maximum ratio $\hat{u}_{i_0}/u_{i_0}$ and
the minimum ratio $\hat{u}_{i_1}/u_{i_1}$, where $\hat{u}_i = e^{\hat{\beta}_i}$ and $u_i = e^{\beta_i}$,
$i_0 = \arg \max_i \hat{u}_i/ u_i$ and  $i_1 = \arg \min_i \hat{u}_i/u_i$.
The first step establishes the existence of an MLE with a high probability.
In the presence of covariates, it is difficult to verify the existence of the MLE.
In addition, it is unclear how to find such neighborhoods because the appearance of covariates will make
some key inequalities in \cite{simons-yao1999} be difficult to generalize.

We use a two-stage process to show consistency.
{\color{black}{
In the first stage, we establish the existence of the MLE and obtain an error bound of $\bm{\widehat{\gamma}}$.
Its proofs proceed three novel inequalities,
including a lower bound of the difference in the expected log-likelihood function $\E \ell(\bm{\theta})$ between $\bm{\theta}$ and its true value $\bm{\theta}^*$,
a uniform bound of the difference in the centered log-likelihood function $\ell(\bm{\theta}) - \E \ell(\bm{\theta})$ between $\bm{\theta}$ and its true value $\bm{\theta}^*$,
and a relationship inequality that connects the lower bound of $\E\ell(\bm{\theta}^*)-\E \ell(\bm{\theta})$ to the error bounds of the MLE.
In this stage, we only obtain a rough inequality for $\bm{\widehat{\beta}}$.
The second stage is to improve this error bound.
We use the convergence rate of the Newton iterative algorithm to solve the equation $F(\mathbf{x})=0$ for
obtaining a much lower bound.
Under the well-known Newton-Kantorovich conditions [\cite{Kantorovich1948Functional}],
the algorithm converges and exhibits a high geometric convergence rate.
As a result, a solution to the equation exists, and an $\ell_p$-error between the initial and limiting points
is obtained.

We now formally state the consistency. The proof of Theorem \ref{Theorem:con} is in the supplementary material.

\begin{theorem}
\label{Theorem:con}
Suppose that Assumptions \ref{assump-m-star}-\ref{assump-residual} hold.
If $b_n^7\kappa_n p_n^{1/2} = o( n^{1/2}/(\log n)^{1/2} )$, then, with a probability of at least $1-O(n^{-1})$,
 the MLE $(\widehat{\bs{\beta}}, \widehat{\bs{\gamma}})$ exists, and satisfies
\begin{equation}
\label{eq-theorem1-beta}
\| \widehat{\bs{\beta}} - \bs{\beta}^* \|_\infty = O\left(  b_n^4 \kappa_n \sqrt{\frac{p_n\log n}{n}} \right), \quad
\| \widehat{\bs{\gamma}} - \bs{\gamma}^{*} \|_2= O\left( b_n\sqrt{\frac{p_n\log n}{n}} \right).
\end{equation}
Further, if the MLE exists, it is unique.
\end{theorem}

}}

We compared our $\ell_\infty$-error bound with \cite{simons-yao1999} in the case of no covariates.
They show that $\| \widehat{\bs{\beta}} - \bs{\beta}\|_\infty = O_p( (\log n/n)^{1/2} )$ when $\|\beta^*\|_\infty$ is bounded by a constant,
our result is consistent with their error bound in \cite{simons-yao1999}, up to some constant factor.
The error bound for $\widehat{\bs{\gamma}}$ is much slower than
the optimal convergence rate $N^{-1/2}$ in classical large-sample theory, where $N=n(n+1)/2$.
The asymptotic distribution result in Theorem \ref{theorem-central-b}
shows that  the convergence rate of $\widehat{\bs{\gamma}}$
is in the order of $O_p(1/n)$, achieving the optimal rate.

We apply the consistency result to the top-$K$ recovery problem, which identifies a set of $K$
subjects with the highest ranks. This problem has received considerable attention in machine learning research; see \cite{chen2019} and references therein.
We assume that there is a ground-truth order $\beta_0^* > \beta_1^* > \cdots > \beta_n^*$.
The aim is to find subjects with $K$ largest estimates   in accordance with their true orders.
It suffices to
demonstrate that
\[
\widehat{\beta}_i - \widehat{\beta}_j > 0,~~i=0, \ldots, K-1; j=K, \ldots, n.
\]
As in \cite{chen2019}, we require a separation measure $\Delta_K = \beta^*_{K-1} - \beta^*_{K}$ to distinguish between the $(K-1)$th and $K$th subjects.
From the triangle inequality, we obtain:
\[
\widehat{\beta}_i - \widehat{\beta}_j \ge \beta_i^* - \beta_j^* - |\widehat{\beta}_i- \beta_i^*|
- |\widehat{\beta}_j-\beta_j^*| \ge \Delta_K - O_p\left( b_n^4\kappa_n \sqrt{\frac{p_n\log n}{n}} \right).
\]
Therefore, we have the following corollary:
\begin{corollary}\label{corollary-topk}
Under the conditions in Theorem \ref{Theorem:con}, if $\Delta_K \gg  b_n^4\kappa_n (p_n\log n)^{1/2}/n^{1/2}$, with a probability of at least $1-O(n^{-1})$, the
set of top-$K$-ranked subjects can be recovered exactly by using MLE.
\end{corollary}

\subsection{Asymptotic normality of $\widehat{\gamma}$}
\label{subsec-asymp-gamma}

Let $\ell_c(\bs{\gamma})$ be the concentrated log-likelihood function of
$\ell(\bs{\beta}, \bs{\gamma})$ by replacing $\bs{\beta}$ with $\widehat{\bs{\beta}}_\gamma$.
It is easy to verify that the Hessian matrix of $-\ell_c(\bs{\gamma})$ (i.e., the Jacobian matrix $Q_c^\prime(\bs{\gamma})$) is
$\Sigma( \bs{\widehat{\beta}}, \bs{\gamma})$, where
\begin{equation}
\label{definition-sigma}
\Sigma(\bs{\beta}, \bs{\gamma}) := \frac{ \partial Q(\bs{\beta}, \bs{\gamma})
}{ \partial \bs{\gamma}^\top} - \frac{ \partial Q(\bs{\beta}, \bs{\gamma})
}{\partial \bs{\beta}^\top} \left[ \frac{\partial H(\bs{\beta}, \bs{\gamma})}{\partial \bs{\beta}^\top} \right]^{-1}
\frac{\partial H(\bs{\beta}, \bs{\gamma})}{\partial \bs{\gamma}^\top}.
\end{equation}
Note that $Q_c^\prime(\bs{\gamma})$ is the Fisher information on $\bs{\gamma}$, which measures
the amount of information on $\bs{\gamma}$ provided by win-loss outcomes.
Therefore, the asymptotic distribution of $\widehat{\bs{\gamma}}$ depends crucially on $Q_c^\prime(\bs{\gamma})$.

Note that \eqref{definition-sigma} involves with the inverse of $\partial H(\bs{\beta}, \bs{\gamma})/\partial \bs{\beta}^\top$, which is denoted as
$H'_\gamma( \bs{\beta} )$ for convenience.
In general, the inverse of $H'_\gamma( \bs{\beta} )$ does not have a closed form. We use a simple matrix
to approximate it.
The Jacobian matrix $H'_\gamma( \bs{\beta} )$ has a special structure that can be characterized in the form of a matrix class.
Given $b_0, b_1>0$, we say that an $n\times n$-matrix $V=(v_{ij})_{i,j=1}^{n}$ belongs to the matrix class $\mathcal{L}_{n}(b_0, b_1)$ if
$V$ is a diagonally dominant matrix with negative nondiagonal elements bounded by $b_0$ and $b_1$; that is,
\begin{equation*}\label{eq1}
\begin{array}{l}
b_0 \le v_{ii} + \sum_{j=1, j\neq i}^{n} v_{ij} \le b_1, ~~i=1,\ldots, n, \\
b_0 \le - v_{ij} \le b_1, ~~ i,j=1,\ldots,n; i\neq j.
\end{array}
\end{equation*}
Define $v_{0n}=v_{n0}=  \sum_{j=1, j\neq i}^{n} v_{ij} -v_{ii}$ for $i=1, \ldots, n$ and $v_{00}=-\sum_{i=1}^n v_{in}$.
\cite{Simons1998Approximating} proposed to approximate the inverse of $V$,  $V^{-1}$, by a simple matrix
$S=(s_{ij})_{n\times n}$, where
\begin{equation}\label{definition-s}
s_{ij} = \frac{\delta_{ij}}{v_{ii}} + \frac{1}{v_{00}}.
\end{equation}
In the above equation, $\delta_{ij}=1$ if $i=j$; otherwise, $\delta_{ij}=0$.
It is clear that $H'_\gamma(\bs{\beta})$ belongs to this matrix class.
Hereafter, we denote $V=(v_{ij})$ by $H'_{\gamma^*}( \bs{\beta}^* )$.

Recall $N=\sum_{i<j}m_{ij}$. Let
\[
\bar{\Sigma}:=\lim_{n\to\infty} \frac{1}{N} \Sigma( \bs{\beta}^*, \bs{\gamma}^*),
\]
where $\Sigma(\bs{\beta}, \bs{\gamma})$ is defined in \eqref{definition-sigma}.
We assume that the limit $\bar{\Sigma}$ exists, which was considered in \cite{Graham2017}.
By using $S$ in \eqref{definition-s} to approximate $V^{-1}$, 
we have
\begin{eqnarray}
\label{eq-approximation-Sigma}
\frac{1}{N}\Sigma( \bs{\beta}^*, \bs{\gamma}^*)  =  \frac{1}{N}\sum\limits_{i<j} \sum\limits_k \mathbf{z}_{ijk}\mathbf{z}_{ijk}^\top \mu^\prime (\pi_{ijk}^*)
- \frac{1}{N}\sum\limits_{i} \frac{ \bm{\bar{Z}}_i \bm{\bar{Z}}_i^\top }{ v_{ii} } + o(1),
\end{eqnarray}
where
\[
\bm{\bar{z}}_i = \sum_{j\neq i}\sum_k z_{ijk}\mu^\prime(\pi_{ijk}^*).
\]
If $\mathbf{z}_{ijk}$ is independently draw from some multivariate distribution, then
$N^{-1}\Sigma( \bs{\beta}^*, \bs{\gamma}^*)$ converges in probability to some non-random matrix.

The idea of establishing  the asymptotic normality of $\widehat{\bs{\gamma}}$ is briefly described as follows:
First, we use the mean-value expansion to derive the explicit expression of $\widehat{\bs{\gamma}}-\bs{\gamma}^*$, which contains
term $Q_c(\bs{\gamma}^*)$ multiplied by $\bar{\Sigma}^{-1}$.
Then, we apply a third-order Taylor expansion to $Q_c(\bs{\gamma}^*)$ to characterize its limiting distribution.
In the expansion, the first-order term is
asymptotically normal; the second-order term is the asymptotic bias term and
the first-order term is the remainder term.
The asymptotic normality of $\widehat{\bs{\gamma}}$ is described as follows, whose proof is in the supplementary material.

\begin{theorem}
\label{theorem-central-b}
{\color{black}{Suppose that the conditions in Theorem \ref{Theorem:con} hold and the dimension $p$ is fixed.
If $b_n^{12}\kappa_n^4 = o( n^{1/2}/(\log n)^{3/2})$}}, then
for a
nonzero constant vector $\mathbf{c}=(c_1, \ldots, c_p)^\top$, $\sqrt{N}\mathbf{c}^\top (\widehat{\bs{\gamma}} - \bs{\gamma})$ converges in distribution to
normal distribution with mean $\bar{\Sigma}^{-1}B_*$ and variance $c^\top \bar{\Sigma} \mathbf{c}$,
\begin{equation}\label{defintion-Bias}
B_*=\lim_{n\to\infty} \frac{1}{2\sqrt{N}} \sum_{i=0}^n \frac{  \sum_{j\neq i} \sum_k \mathbf{z}_{ijk} \mu^{\prime\prime}(\pi^*_{ijk})  }
{  \sum_{j\neq i} \sum_k \mu^\prime(\pi^*_{ijk})  }.
\end{equation}
\end{theorem}

\begin{remark}
\label{remark-gamma}
If $\lambda_{\min}( \Sigma( \bs{\beta}^*, \bs{\gamma}^*) ) \ge c_0 n^2$, then $\widehat{\bs{\gamma}}$ has a convergence rate $O(n^{-1})$.
The asymptotic distribution of $\widehat{\bs{\gamma}}$ contains a bias term $B_*$.
This is because of the different convergence rates of $\widehat{\bs{\gamma}}$ and $\widehat{\bs{\beta}}$,
which roughly are $O_p(1/n)$ and  $O_p(1/n^{1/2})$.
This phenomenon is referred to as the incidental parameter problem; see
econometric literature \cite{Graham2017} and the references therein.
The bias term is bounded above by $b_n\kappa_n$. This is due to that $\mu^\prime(\pi^*_{ijk})\ge 1/b_n$
and $|\mu^{\prime\prime}(\pi^*_{ijk})|\le 1/4$.
If $B_*\neq 0$, then we can use the empirically analytical bias-correction formula as in \cite{Dzemski2019}:
$\widehat{\bs{\gamma}}_{bc} = \widehat{\bs{\gamma}}- N^{-1/2}\widehat{\Sigma}^{-1}(\widehat{\bs{\beta}},\widehat{\bs{\gamma}}) \hat{B}$,
where $\widehat{B}$ and $\widehat{\Sigma}$ are the estimates of $B_*$ and $\bar{\Sigma}$ obtained by replacing
$\bs{\beta}^*$ and $\bs{\gamma}^*$ in their expressions with the estimators $\widehat{\bs{\beta}}$ and
$\widehat{\bs{\gamma}}$.
\end{remark}

\subsection{Asymptotic normality of $\widehat{\beta}$}

The idea of establishing an asymptotic distribution $\widehat{\bs{\beta}}$ is briefly described as follows.
A second-order Taylor expansion is applied to $H_\gamma( \widehat{\bs{\beta}} )$ at $\bs{\beta}$
to derive the following explicit asymptotic expression for $\widehat{\bs{\beta}}$.
In the expansion, the first-order term is the sum of $[H'_\gamma( \widehat{\bs{\beta}} )]^{-1}(\widehat{\bs{\beta}} - \bs{\beta})$
and $V_{\gamma\beta}(\widehat{\bs{\gamma}}-\bs{\gamma})$, where $V_{\gamma\beta}=\partial H(\bs{\beta},\bs{\gamma})/\partial \bs{\gamma}^\top$.
Because $[H'_\gamma( {\bs{\beta}} )]^{-1}$ does not have a closed form, we use $S$ defined in \eqref{definition-s} to approximate it.
As discussed in Remark \ref{remark-gamma}, $\lambda_{\min}( \Sigma( \bs{\beta}^*, \bs{\gamma}^*) ) \ge c_0 n^2$ guarantees that $\widehat{\bs{\gamma}}$ has
a convergence rate of $O_p(1/n)$.
This makes that the term $V_{\gamma\beta}(\widehat{\bs{\gamma}}-\bs{\gamma})$ is an asymptotically neglected remainder term.
The second-order term in the expansion is also asymptotically neglected.
Then, we represent $\widehat{\bs{\beta}}-\bs{\beta}$ as the sum of
$S(\mathbf{d} - \E \mathbf{d})$ and remaining terms, where $\mathbf{d}=(d_1, \ldots, d_n)^\top$.
Therefore, the central limit theorem is proven by establishing the asymptotic normality of $S( \mathbf{d} - \E \mathbf{d})$ and
indicating that the remaining terms are negligible.
We formally state the central limit theorem below, whose proof is in the supplementary material.

\begin{theorem}\label{Theorem-central-a}
{\color{black}{
Suppose that the conditions in Theorem \ref{Theorem:con} hold, $\lambda_{\min}( \Sigma( \bs{\beta}^*, \bs{\gamma}^*) ) \ge c_0 n^2$ and
the dimension $p$ is fixed.}}
If $b_n^{12}\kappa_n^4 = o( n^{1/2}/\log n)$, then, for a fixed $k$ the vectors $( (\widehat{\beta}_1 - \beta^*_1), \ldots,  (\widehat{\beta}_k - \beta^*_k))$
follows a $k$-dimensional multivariate normal distribution with a covariance matrix given by the upper left $k\times k$ block of $S$
defined in \eqref{definition-s}.
\end{theorem}

\begin{remark}
The asymptotic variance of $\widehat{\beta_i}$ is $1/v_{ii} + 1/v_{00}$, which is in the magnitudes of $O(n^{1/2})$.
In case of no covariates, it is consistent with that in \cite{simons-yao1999}.
\end{remark}

\section{Extensions}
\subsection{Extension to an Erd\H{o}s--R\'{e}nyi comparison graph}
\label{section:extension}

All the preceding results concern dense comparisons, where each pair has at least one comparison.
We extended these to an Erd\H{o}s--R\'{e}nyi comparison graph $\mathcal{G}(n, q_n)$, where any two subjects are compared with probability $q_n$.
If $q_n \to 0$, this implies a sparse comparison design.
We assume that if two subjects are compared, they are compared at most $m_*$ times with $m_*$ fixed,
according to the aforementioned settings.
In addition, we consider the case of an increasing dimension of covariates, i.e., $p_n\to\infty$.
The consistency result is stated below, whose proof is in the supplementary material.

\begin{theorem}
\label{Theorem-con-incr}
Assume that $\| \bs{\beta}^* \|_\infty$ and $\|\bs{\gamma}^*\|_2$ are bounded by a positive constant,
$\kappa=\sup_{i,j,k} \| \mathbf{z}_{ijk} \|_2\le C$ for some constant $C$,
and $q_n \ge c \log n/n$ for a sufficiently large constant $c$.
Under Assumptions \ref{assump-m-star}, \ref{assump-a-indepen} and \ref{assump-residual}, if $p_n/q_n=o(n/\log n)$
then, with a probability of at least $1-O(n^{-1})$,
 the MLE $(\widehat{\bs{\beta}}, \widehat{\bs{\gamma}})$ exists and satisfies
\begin{equation*}
\| \widehat{\bs{\beta}} - \bs{\beta}^* \|_\infty = O\left(  \sqrt{\frac{\log n}{ nq_n }} \right), ~~
\| \widehat{\bs{\gamma}} - \bs{\gamma}^{*} \|_2= O\left( \sqrt{\frac{p_n\log n}{nq_n}} \right).
\end{equation*}
\end{theorem}

When sampling probability $q_n$ is less than $\log n/n$, the realized comparison graph is disconnected with a positive probability
according to the theory of the Erd\H{o}s--R\'{e}nyi graph. In this case, all subjects can be divided into two groups such that
any subject in the first group does not have comparisons with any subject in the second group, where it is not possible to give
a ranking of all subjects.
Therefore, $q_n$ should be not smaller than $\log n/n$, up to a constant factor, which is a fundamental
requirement to guarantee the connection of the sampling graph.
Condition $p_n/q_n=o(n/\log n)$ restricts the increasing rate of $p_n$.
When $q_n$ is a constant, it is weaker than that in \cite{wang2011}, where $p_n^2=o(n/\log n)$.

\subsection{Extensions to a fixed sparse comparison graph with a dynamic range of merit parameters}
\label{section:extension-fixed}

We extend them to a fixed sparse comparison graph in \cite{yan2012sparse} here.
In some applications such as sports, the comparison graph may be fixed, not be random.
For example, in the regular season of the National Football League (NFL),  which teams having games are scheduled in advance.
More specially, there are $32$ teams in the two conferences of the NFL and are divided into eight divisions each consisting of four teams.
In the regular season, each team plays $16$ matches, $6$ within the division and $10$
between the divisions.
Motivated by the design of the regular season of the National Football League,  they proposed a sparse condition to control the length from one
subject to another subject with $2$ or $3$ as in \cite{yan2012sparse}:
\begin{equation}\label{eq-difintion-taun}
\tau_n := \min_{0\le i<j \le n} \frac{\#\{k: m_{ik}>0, m_{jk}>0 \} }{n}.
\end{equation}
That is, $\tau_n$ is the minimum ratio of the total number of paths between any $i$ and $j$ with length $2$ or $3$.

We assume that if two subjects have comparisons, they are compared at most $m_*$ times with $m_*$ fixed,
in accordance with the aforementioned setting.
The same proof technique can be readily extended to the setting here.
The main different places are the error bound of using $S$ to approximate $V^{-1}$, $\|V^{-1} - S \|_{\max}$, and the number of comparisons of subject $i$, $m_i$,  that will be replaced with $b_n^3/(n^2 \tau_n^3)$ and $n\tau_n$ in the sparse case.
Here, $V=H^\prime_{\gamma^*}(\bs{\beta}^*)$.
In this section, I redefine $b_n$ in Assumption \ref{assump-parameter} as
\begin{equation}\label{eq-definition-bn}
b_{n} := \max_{ i,j, k} \frac{ (1+e^{\pi_{ijk}^*})^2 }{ e^{\pi_{ijk}^*}  } = =O( e^{ \max_{i,j}(\beta_i^*-\beta_j^*) + z_*\|\bs{\gamma}^*\|_1 }),
\end{equation}
where $z_*=\max_{i,j,k} \|Z_{ijk}\|_\infty$.
That says $\min_{i,j,k} \mu^\prime( \pi_{ijk}^*) \ge 1/b_n$. It is easy to see $b_n\ge 4$.

Let $\lambda_{\min}(\bs{\beta})$ be the smallest eigenvalue of $n^{-2}\Sigma(\bs{\beta}, \bs{\gamma}^*)$ and define
\begin{equation}\label{definition-rhon}
\rho_n := \sup_{\bs{\beta}\in B(\bs{\beta}^*, \epsilon_{n1})} \frac{\sqrt{2}}{\lambda_{\min}(\bs{\beta})}.
\end{equation}
Let $\| A \|_2$ be the $\ell_2$-norm of a matrix $A$ induced by Euclidean norm on vectors.
By the inequality of matrix norm, as in \cite{Matrixnorm} (p. 56--57),  we have
\begin{equation}\label{Omega-inverse-bound}
\sup_{\bs{\beta}\in \Sigma(\bs{\beta}^*, \epsilon_{n1})} \| \Sigma^{-1}(\bs{\beta}, \bs{\gamma}^*)\|_\infty
\le \sup_{\bs{\beta}\in \Sigma(\bs{\beta}^*, \epsilon_{n1})} \sqrt{2}\| \Sigma^{-1}(\bs{\beta}, \bs{\gamma}^*)\|_2
 \le  \frac{ \rho_n }{ n^2}.
\end{equation}
Note that the dimension of the matrix $\Sigma(\bs{\beta}, \bs{\gamma})$ is fixed and every its entry is a sum of $n(n-1)/2$ terms.
There it is suitable to have a factor $n^{-2}$ in the above inequality.
We have the following theorem.  

\begin{theorem}\label{Theorem:fixed}
Suppose that $\max_{i,j,k}\|\bm{z}_{ijk}\|_\infty\le c$ for a fixed constant $c>0$.
(1) If $\rho_n b_n^{9}/\tau_n^9 = o( (n/\log n)^{1/2})$, then with probability at least $1-O(n^{-1})$,
 the MLE $(\widehat{\bs{\beta}}, \widehat{\bs{\gamma}})$ exists and satisfies
\begin{equation}\label{th-Newton-convergence-rate}
\| \widehat{\bs{\gamma}} - \bs{\gamma}^{*} \|_\infty = O_p\left(  \frac{\rho_n b_{n}^{9} \log n }{ n \tau_n^9}) \right)=o_p(1),~~
\| \widehat{\bs{\beta}} - \bs{\beta}^* \|_\infty = O_p\left( \frac{b_n^3}{\tau_n^3} \sqrt{\frac{\log n}{n}} \right)=o_p(1).
\end{equation}
(2) 
If  $\rho_n b_n^9/\tau_n^9 = o( n^{1/2}/(\log n)^{1/2} )$,
then for fixed $k$, the vector $( (\widehat{\beta}_1 - \beta^*_1), \ldots,  (\widehat{\beta}_k - \beta^*_k))$
follows a $k$-dimensional multivariate normal distribution with mean zero and the covariance matrix  given by the upper left $k\times k$ block of $S$
defined at \eqref{definition-s}. \\
(3) If $ b_{n}/\tau_n = o( n^{1/24}/(\log n)^{/24} )$ and $\rho_n b_n^{9}/\tau_n^9 = o( (n/\log n)^{1/2})$, then
for arbitrarily given nonzero constant vector $c=(c_1, \ldots, c_p)^\top$, $\sqrt{N}c^\top (\widehat{\gamma} - \gamma)$ converges in distribution to
the normal distribution with mean $\bar{\Sigma}^{-1}B_*$ and variance $c^\top \bar{\Sigma} c$.
\end{theorem}

The proof of claim (1) in Theorem \ref{Theorem:fixed} is in the supplementary material.
The proofs of claims (2) and (3) in Theorem \ref{Theorem:fixed} are respectively similar to the proofs of Theorem \ref{theorem-central-b}
and Theorem \ref{Theorem-central-a}, and omitted.
Note that in the proof of the consistency result in Theorem \ref{Theorem:fixed},
the independence assumption is only used to derive concentration inequalities
for $\sum_{j\neq i}\sum_k (A_{ijk} - \E A_{ijk})$ and $\sum_{i<j}\sum_k \bm{z}_{ijk} (A_{ijk} - \E A_{ijk})$ in Lemma \ref{lemma-consistency-beta-b}.
This result holds, regardless of the independence assumption, as long as the following event $E_n$ holds with high probability:
\begin{equation}
\label{definition-En-a}
E_n =  \left\{  \max_{i} |H_i(\bs{\beta}^*, \bs{\gamma}^*) | = O( m_{\max}\log m_{\max})^{1/2}),
~~\|Q(\bs{\beta}^*, \bs{\gamma}^*) \|_\infty = O(  N\log N ) \right\}.
\end{equation}
Therefore, we have the following corollary.

{\color{black}{
\begin{corollary}
\label{coro-a}
Under Assumption \ref{assump-m-star}, if $\P(E_n)\to 1$, $\max_{i,j,k}\|\bm{z}_{ijk}\|_\infty\le c$, and  $\rho_n b_n^{9}/\tau_n^9 = o( (n/\log n)^{1/2})$,
 the MLE $(\widehat{\bs{\beta}}, \widehat{\bs{\gamma}})$ exists and satisfies
\eqref{th-Newton-convergence-rate}.
\end{corollary}
}}

\section{Numerical Studies}
\label{section:simulation}

In this section, we evaluate the asymptotic results of the MLE using simulation studies and a real-world data example.

\subsection{Simulation studies}

{\color{black}{
\subsubsection{Setting: fixed $p=2$, moderate $n$ and $m_{ij}=1$}
\label{subsubsec-moderate}

We assume that each subject is compared with another subject only once, that is, $m_{ij}=1$ for all $i\neq j$.
A comparison between $i$ and $j$ is associated with the two-dimensional covariate vector $\mathbf{z}_{ij1}=(z_{ij1,1}, z_{ij1,2})^\top$.
We consider two cases (Case I) and (Case II) for the covariate settings:
\begin{itemize}
\item (Case I).
when $i<j$, $z_{ij1,1}$ took values $-1$ or $1$ randomly with equal probability, and $z_{ij1,2}$
was generated using a standard normal distribution.

\item
(Case II).
when $i<j$, $z_{ij1,1}$ takes values $-1$ or $1$ randomly with respective probability $0.3$ and $0.7$, and $z_{ij1,2}$ takes values $-1$ or $1$ randomly with respective probability $0.2$ and $0.8$.
\end{itemize}
Note that $\mathbf{z}_{ji1}=-\mathbf{z}_{ij1}$.
All covariates were generated independently.
We set the coefficient regression of $\mathbf{z}_{ij1}$: $\gamma^*=(0.3, 0.3)^\top$.
The merit parameters took a linear form, i.e.,
$\beta_i^* = ic\log n/n$ for $i=0, \ldots, n$, where $\max_{i,j}(\beta_i^* - \beta_j^*)= c\log n$.
To assess asymptotic properties under different asymptotic regimes,
we considered four different values of $c$: $c=0, 0.05, 0.1, 0.2$ and two different values of $n$:
$n=100, 200$.
Each simulation was repeated $10,000$ times.

From Theorem \ref{Theorem-central-a}, both $\hat{\eta}_i=(\hat{\beta}_i - \beta_i^*)/(1/\hat{v}_{ii}+1/\hat{v}_{00})^{1/2}$ and
$\hat{\xi}_{i,j} = [\hat{\beta}_i-\hat{\beta}_j-(\beta_i^*-\beta_j^*)]/(1/\hat{v}_{ii}+1/\hat{v}_{jj})^{1/2}$
converges in distribution to the standard normalities, where $\hat{v}_{i,i}$ is the estimate of $v_{i,i}$
by replacing $(\bs{\beta}^*, \bs{\gamma}^*)$ with $(\widehat{\bs{\beta}}, \widehat{\bs{\gamma}})$.
Therefore, we assessed the asymptotic normalities of $\hat{\eta}_i$ and $\hat{\xi}_{i,j}$ using the quantile-quantile (QQ) plot. The QQ plots for $\hat{\eta}_i$ and $\hat{\xi}_{i,j}$ under both cases (I) and (II) are similar.
We only show the plots for $\hat{\eta}_i$ in Figure \ref{fig-non-100-beta-i} in case (II) when $n=100$. From this figure, we can see that the QQ plots align with the diagonal lines very well in each sub-figure, indicating the asymptotic normal approximation is good.

\begin{figure}[!htbp] 
    \centering
    \includegraphics[width=0.8\textwidth]{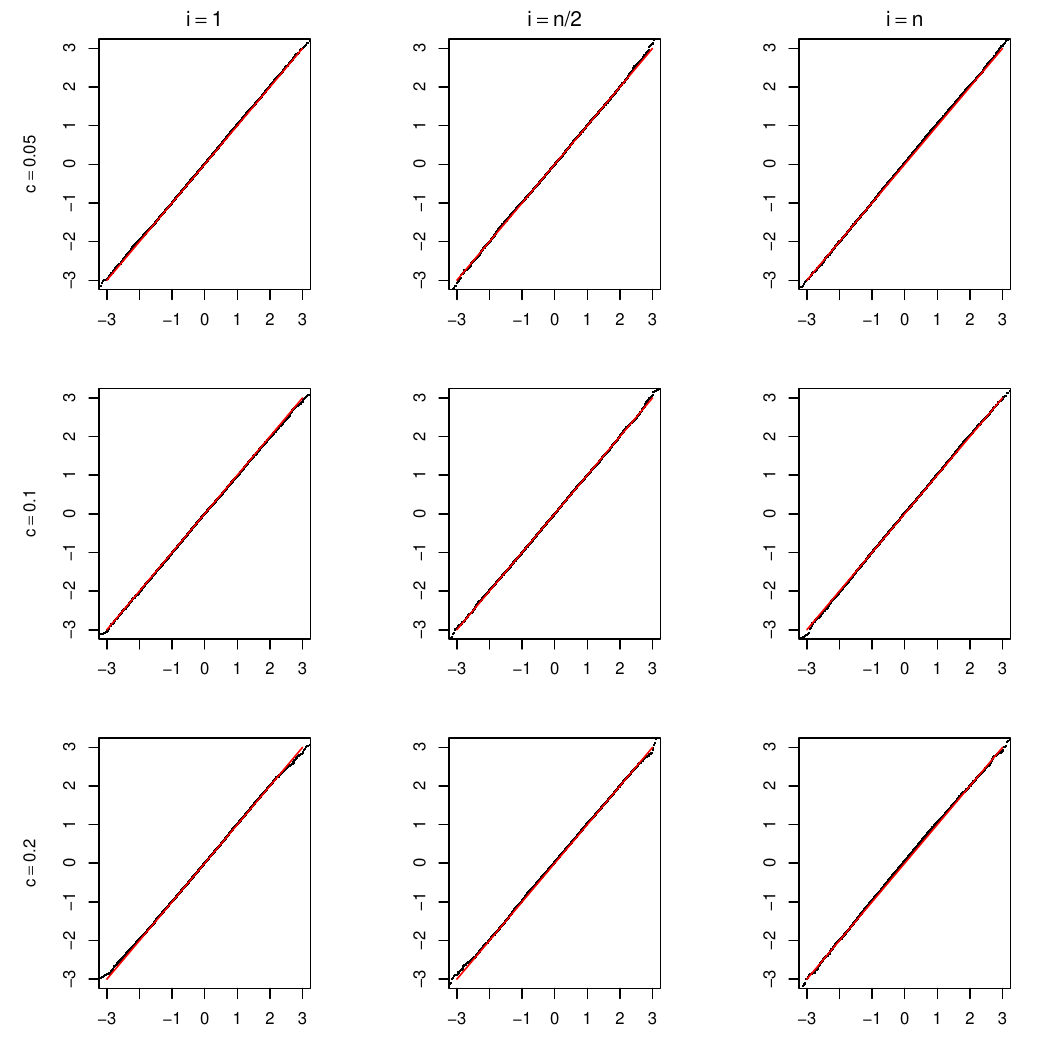} %
    \caption{The figures are the quantile-quatile plots of $(\hat{\beta}_i - \beta_i)/(1/\hat{v}_{ii} + 1/\hat{v}_{00})^{1/2}$
     against the standard normality under the setting Case II,
    where the horizontal and vertical coordinates denote the theoretical
    and sample quantiles, respectively.
    The number of subjects is $n=100$.
    The first, second, third rows correspond to $c=0.05$, $c=0.1$ and $c=0.2$ in the $\beta$-parameter settings.
     The plot under the case of $c=0$ is similar to that under $c=0.05$ and is not shown to save space.
     The first and second columns correspond to $i=1$, $i=n/2$ and $i=n$, respectively. } %
    \label{fig-non-100-beta-i}
\end{figure}

\begin{table}[!htpb]\centering
\caption{The reported values are the coverage frequency ($\times 100\%$) for $\beta_i$ for a pair $i$, denoted as ``CF" and the length of the confidence interval ``LCI".}
\label{tab-moderate-beta}
\footnotesize
\vspace{0.2cm}
\renewcommand{\arraystretch}{1.2}
\begin{tabular}{ c c c cc c cc c cc c cc}
\hline
 n& &&\multicolumn{2}{c}{$c=0$}
&&\multicolumn{2}{c}{$c=0.05$}
&&\multicolumn{2}{c}{$c=0.1$}
&&\multicolumn{2}{c}{$c=0.2$}
\\
\cmidrule(r){4-5} \cmidrule(r){7-8} \cmidrule(r){10-11} \cmidrule(r){13-14}
  &&&  CF&LCI  && CF&LCI
&&  CF&LCI  && CF&LCI
\\
\midrule
$100$&&& \multicolumn{11}{c}{Case I: $z_{ijk,1}\in \{-1,1\}$ with equal prob., $z_{ijk,2}\sim N(0,1)$}
\\
\cline{4-14}
&$1$ && $ 94.99 $&$ 1.14 $&&$ 94.71 $&$ 1.15 $&&$ 94.88 $&$ 1.15 $&&$ 95.15 $&$ 1.18 $
\\
&$50$ &&$ 94.99 $&$ 1.14 $&&$ 94.71 $&$ 1.15 $&&$ 94.61 $&$ 1.15 $&&$ 95.15 $&$ 1.17 $
\\
&$100$ &&$ 94.5 $&$ 1.14 $&&$ 94.59 $&$ 1.15 $&&$ 94.69 $&$ 1.15 $&&$ 94.85 $&$ 1.18 $
\\
$200$&$1$ &&$ 94.71 $&$ 0.81 $&&$ 95.01 $&$ 0.81 $&&$ 95.05 $&$ 0.81 $&&$ 95.37 $&$ 0.84 $
\\
&$100$ &&$ 94.69 $&$ 0.81 $&&$ 95.08 $&$ 0.81 $&&$ 95.2 $&$ 0.81 $&&$ 95.12 $&$ 0.83 $
\\
&$200$ &&$ 94.92 $&$ 0.81 $&&$ 95.29 $&$ 0.81 $&&$ 95.15 $&$ 0.81 $&&$ 95.15 $&$ 0.84 $
\\
\midrule
&&& \multicolumn{11}{c}{Case II: $z_{ijk,1}\in \{-1,1\}$ with $(0.3,0.7)$,
$z_{ijk,2}\in \{-1,1\}$ with $(0.2,0.8)$}
\\
\cline{4-14}
100&$1$ &&$ 95.09 $&$ 1.15 $&&$ 95.09 $&$ 1.16 $&&$ 95.03 $&$ 1.18 $&&$ 94.99 $&$ 1.23 $
\\
&$50$ &&$ 94.96 $&$ 1.15 $&&$ 94.92 $&$ 1.16 $&&$ 95.06 $&$ 1.17 $&&$ 95.09 $&$ 1.2 $
\\
&$100$ &&$ 94.47 $&$ 1.15 $&&$ 94.61 $&$ 1.16 $&&$ 94.68 $&$ 1.18 $&&$ 95.01 $&$ 1.23 $
\\
200&$1$ &&$ 95.07 $&$ 0.81 $&&$ 95.16 $&$ 0.82 $&&$ 95.2 $&$ 0.83 $&&$ 95.09 $&$ 0.87 $
\\
&$100$ &&$ 95.06 $&$ 0.81 $&&$ 94.98 $&$ 0.82 $&&$ 95.29 $&$ 0.83 $&&$ 94.8 $&$ 0.85 $
\\
&$200$ &&$ 94.92 $&$ 0.81 $&&$ 94.88 $&$ 0.82 $&&$ 94.98 $&$ 0.83 $&&$ 94.59 $&$ 0.87 $
\\
\bottomrule
\end{tabular}
\end{table}

The MLE existed in all the simulations.
Table \ref{tab-moderate-beta} reports the $95\%$  coverage frequencies for $\beta_i^*$ and the length of the confidence interval.
As we can see, the length of the confidence interval decreases as $n$ increases, which qualitatively agrees with the theory.
Because the difference in the merit parameters between adjacent subjects was very small, the lengths of the confidence intervals were very close across different pairs.
This is consistent with the theoretical length.
All the simulated coverage frequencies are close to the nominal level $95\%$, indicating the asymptotic approximation performs well.

We draw the QQ plots for the uncorrected estimator $\bm{\widehat{\gamma}}$ and
its bias-corrected estimator  $\bm{\widehat{\gamma}}_{corr}$.
Since there are similar phenomena under two different covariate setting and two different sizes
$n=100$ and $n=200$, only the  QQ plots under  Case I and $n=100$ is shown in  Figure
\ref{fig-uniform-100-gamma} and other plots are not presented.
From this figure, we can see that the quantile-quatile plots of $(\hat{\gamma}_1 - \gamma_1)/\hat{\sigma}_1$
     against the standard normality
deviate obviously upward from the diagonal line
while the QQ plots for the bias-corrected estimator
$\hat{\gamma}_{i,corr}$ agree with the diagonal lines well.
Here, $\hat{\sigma}_1$ denotes the MLE of the standard error of $\hat{\gamma}_1$.
There are similar phenomena in the QQ plots for $\hat{\gamma}_2$, which are not shown to save space.
This shows that the bias-corrected estimator leads to
more accurate statistical inference than the uncorrected estimator for the regression coefficient $\bm{\gamma}$.

Table \ref{tab1} reports
 the coverage frequency ($\times 100\%$) for $\widehat{\bs{\gamma}}$
(CP) and the coverage frequency ($\times 100\%$) for $\widehat{\bs{\gamma}}_{bc}$ ($\mathrm{CP}_{corr}$).
From this table, we can see that (1) both values are close to the target coverage
frequency $95\%$; (2)
the coverage frequency for the uncorrected estimator
is a littler lower than $95\%$; (3)
the coverage frequencies for the corrected estimator are closer to the target level than
those for the uncorrected estimator.

\begin{figure}[!htbp] 
    \centering
    \includegraphics[width=0.8\textwidth]{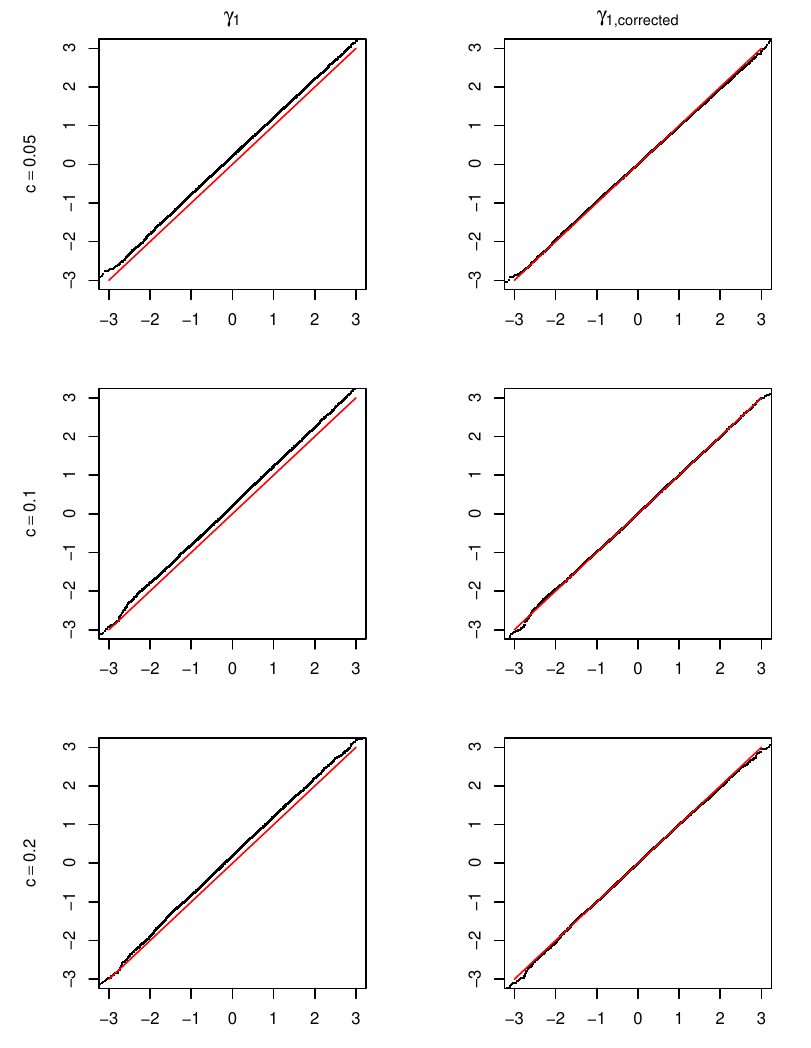} %
    \caption{The figures are the quantile-quatile plots of $(\hat{\gamma}_1 - \gamma_1)/\hat{\sigma}_1$
     against the standard normality under the setting Case I ($z_{ij,1}\in \{-1,1\}$ uniformly, $z_{ij,2}\sim N(0,1)$),
    where the horizontal and vertical coordinates denote the theoretical
    and sample quantiles, respectively.
    Here, the number of items is $n=100$.
    The first, second, third rows correspond to $c=0.05$, $c=0.1$ and $c=0.2$ in the $\beta$-parameter settings.
     The plot under the case of $c=0$ is similar to that under $c=0.05$ and is not shown to save space.
     The first and second columns represent the uncorrected and corrected statistics for $\gamma_1$, where $\hat{\gamma}_1$ denotes the MLE of the standard error of $\hat{\gamma}_1$. } %
    \label{fig-uniform-100-gamma}
\end{figure}

\begin{table}[!htpb]\centering
\caption{The reported values are the coverage frequency ($\times 100\%$) for $\widehat{\bs{\gamma}}$
(CP), the coverage frequency ($\times 100\%$) for $\widehat{\bs{\gamma}}_{bc}$ ($\mathrm{CP}_{corr}$), and
length of confidence interval (LCI).}
\label{tab1}
\scriptsize
\renewcommand{\arraystretch}{1.2}
\vspace{0.2cm}
\begin{tabular}{c c c ccc c ccc c ccc c ccc}
\hline
& &&\multicolumn{3}{c}{$c=0$}
 &&\multicolumn{3}{c}{$c=0.05$}
 &&\multicolumn{3}{c}{$c=0.1$}
&&\multicolumn{3}{c}{$c=0.2$}
\\
\cmidrule(r){4-6} \cmidrule(r){8-10} \cmidrule(r){12-14} \cmidrule(r){16-18}
$n$  &$\gamma$&&  CP&$\mathrm{CP}_{corr}$&LCI  && CP&$\mathrm{CP}_{corr}$&LCI
&&  CP&$\mathrm{CP}_{corr}$&LCI  && CP&$\mathrm{CP}_{corr}$&LCI
\\
\midrule
&&& \multicolumn{15}{c}{Case I: $z_{ij1}\in \{-1,1\}$ uniformly,~~ $z_{ij2}\sim N(0,1)$ }
\\
\cline{4-18}
$100$ & $\gamma_1$ && $ 94.0 $&$ 95.1 $&$ 0.12 $&&$ 94.45 $&$ 95.37 $&$ 0.12 $&&$ 94.34 $&$ 95.41 $&$ 0.12 $&&$ 94.29 $&$ 94.92 $&$ 0.12 $
\\
  & $\gamma_2$ &&$ 93.94 $&$ 94.9 $&$ 0.12 $&&$ 93.95 $&$ 95.05 $&$ 0.12 $&&$ 93.66 $&$ 94.73 $&$ 0.12 $&&$ 94.27 $&$ 95.05 $&$ 0.12 $
\\
$200$ & $\gamma_1$ && $ 94.15 $&$ 94.89 $&$ 0.06 $&&$ 94.51 $&$ 95.13 $&$ 0.06 $&&$ 94.37 $&$ 94.9 $&$ 0.06 $&&$ 94.71 $&$ 95.33 $&$ 0.06 $
\\
 & $\gamma_2$ && $ 94.84 $&$ 95.31 $&$ 0.06 $&&$ 94.33 $&$ 94.91 $&$ 0.06 $&&$ 94.62 $&$ 95.23 $&$ 0.06 $&&$ 94.56 $&$ 95.24 $&$ 0.06 $
 \\
 \midrule
&&& \multicolumn{15}{c}{Case II: $z_{ij1}\in \{-1,1\}$ with probabilities $0.3$ and $0.7$,
$z_{ij2}\in \{-1,1\}$ with probabilities $0.2$ and $0.8$. }
\\
\cline{4-18}
$100$ & $\gamma_1$ && $ 94.4 $&$ 95.11 $&$ 0.12 $&&$ 94.44 $&$ 95.09 $&$ 0.12 $&&$ 94.44 $&$ 95.16 $&$ 0.12 $&&$ 94.33 $&$ 95.01 $&$ 0.12 $
\\
  & $\gamma_2$ && $ 94.16 $&$ 94.85 $&$ 0.12 $&&$ 94.32 $&$ 95.29 $&$ 0.12 $&&$ 94.39 $&$ 94.85 $&$ 0.12 $&&$ 94.39 $&$ 94.77 $&$ 0.12 $
\\
$200$ & $\gamma_1$ && $ 93.83 $&$ 94.54 $&$ 0.06 $&&$ 94.54 $&$ 95.16 $&$ 0.06 $&&$ 94.85 $&$ 95.39 $&$ 0.06 $&&$ 94.52 $&$ 95.03 $&$ 0.06 $
\\
 & $\gamma_2$ && $ 94.8 $&$ 95.17 $&$ 0.06 $&&$ 94.14 $&$ 94.59 $&$ 0.06 $&&$ 94.5 $&$ 95.03 $&$ 0.06 $&&$ 93.97 $&$ 94.17 $&$ 0.06 $
\\
\bottomrule
\end{tabular}
\end{table}

\subsubsection{Setting: fixed $p=2$, small $n=30$ and $m_{ij}=4$}

Motivated by the schedule of the NBA regular
season, where each team plays against every other team $3$ or $4$ times in the regular season,
we carried out simulations under a small size $n=30$ and the number of repeated comparisons
$m_{ij}=4$ for all $i\neq j$.
The parameter and covariate settings are the same as those in section \ref{subsubsec-moderate}.
Each simulation was repeated $10,000$ times.

All the QQ plots for $\hat{\eta}_{i}$ and $\hat{\eta}_{ij}$, like those QQ plots
in Figure \ref{fig-non-100-beta-i},
 match the diagonal lines very well. Thus, we do not show them here to save space.
Because of the same reason, I only show  the QQ-plots for $\widehat{\gamma}_2$
and  $\widehat{\gamma}_{2,corr}$ under Case II in Figure \ref{fig-non-30-gamma-2}.
From this figure,  we can see that both
 quantile-quatile plots of $(\hat{\gamma}_2 - \gamma_2)/\hat{\sigma}_2$ and its corrected version $(\hat{\gamma}_{2,corr} - \gamma_2)/\hat{\sigma}_2$
     against the standard normality
agree well with the diagonal line, although there is a very small derivations for the uncorrected estimator.
In contrast to Figure \ref{fig-uniform-100-gamma},
we find that increasing the number of comparisons in each pair can effectively reduce the bias.
These results indicate that the asymptotic normal approximation is good even if $n$ is small.

The coverage frequencies and the  confidence intervals with the target level $95\%$
are in Table \ref{tab-small-beta} for $\beta$ parameter and Table \ref{tab-small-gamma} for $\gamma$ parameter, respectively.
From these tables, we can see that all simulated coverage frequencies are close to the target level.

\begin{table}[!htpb]\centering
\caption{The reported values are the coverage frequency ($\times 100\%$) for $\beta_i-\beta_j$ for a pair $(i,j)$, denoted as ``CF" and the length of the confidence interval ``LCI", where the number of items is $30$.}
\label{tab-small-beta}
\footnotesize
\vspace{0.2cm}
\renewcommand{\arraystretch}{1.2}
\begin{tabular}{ c c cc c cc c cc c cc}
\hline
 &&\multicolumn{2}{c}{$c=0$}
&&\multicolumn{2}{c}{$c=0.05$}
&&\multicolumn{2}{c}{$c=0.1$}
&&\multicolumn{2}{c}{$c=0.2$}
\\
\cmidrule(r){3-4} \cmidrule(r){6-7} \cmidrule(r){9-10} \cmidrule(r){12-13}
  &&  CF&LCI  && CF&LCI
&&  CF&LCI  && CF&LCI
\\
\midrule
&& \multicolumn{11}{c}{Case I: $z_{ijk,1}\in \{-1,1\}$ with equal prob., $z_{ijk,2}\sim N(0,1)$}
\\
\cline{3-13}
$(1,2)$ && $ 94.88 $&$ 1.03 $&&$ 95.26 $&$ 1.04 $&&$ 94.48 $&$ 1.04 $&&$ 95.05 $&$ 1.05 $
\\
$(15,16)$ &&$ 95.28 $&$ 1.03 $&&$ 94.78 $&$ 1.03 $&&$ 94.76 $&$ 1.04 $&&$ 94.98 $&$ 1.04 $
\\
$(29,30)$ &&$ 94.81 $&$ 1.03 $&&$ 95.06 $&$ 1.04 $&&$ 94.49 $&$ 1.04 $&&$ 94.95 $&$ 1.05 $
\\
$(0,1)$ &&$ 94.67 $&$ 1.03 $&&$ 94.99 $&$ 1.04 $&&$ 95.05 $&$ 1.04 $&&$ 95.02 $&$ 1.05 $
\\
$(0,15)$ &&$ 95.2 $&$ 1.03 $&&$ 95.26 $&$ 1.03 $&&$ 94.61 $&$ 1.04 $&&$ 95.08 $&$ 1.05 $
\\
$(0, 30)$ &&$ 95.05 $&$ 1.03 $&&$ 95.13 $&$ 1.04 $&&$ 94.66 $&$ 1.04 $&&$ 95.11 $&$ 1.05 $
\\
\midrule
&& \multicolumn{11}{c}{Case II: $z_{ijk,1}\in \{-1,1\}$ with $(0.3,0.7)$,
$z_{ijk,2}\in \{-1,1\}$ with $(0.2,0.8)$}
\\
\cline{3-13}
$(1,2)$ &&$ 94.37 $&$ 1.04 $&&$ 95.01 $&$ 1.05 $&&$ 94.76 $&$ 1.06 $&&$ 94.94 $&$ 1.08 $
\\
$(15,16)$ &&$ 95.05 $&$ 1.04 $&&$ 94.89 $&$ 1.04 $&&$ 94.96 $&$ 1.05 $&&$ 94.91 $&$ 1.06 $
\\
$(29,30)$ &&$ 95.18 $&$ 1.04 $&&$ 94.99 $&$ 1.05 $&&$ 94.83 $&$ 1.06 $&&$ 94.76 $&$ 1.08 $
\\
$(0,1)$ &&$ 94.98 $&$ 1.04 $&&$ 95.12 $&$ 1.05 $&&$ 94.82 $&$ 1.06 $&&$ 95.03 $&$ 1.08 $
\\
$(0,15)$ &&$ 94.61 $&$ 1.04 $&&$ 94.86 $&$ 1.05 $&&$ 94.85 $&$ 1.05 $&&$ 94.75 $&$ 1.07 $
\\
$(0, 30)$ &&$ 94.03 $&$ 1.04 $&&$ 93.93 $&$ 1.05 $&&$ 94.26 $&$ 1.06 $&&$ 93.5 $&$ 1.08 $
\\
\bottomrule
\end{tabular}
\end{table}

\begin{figure}[!htbp] 
    \centering
    \includegraphics[width=0.8\textwidth]{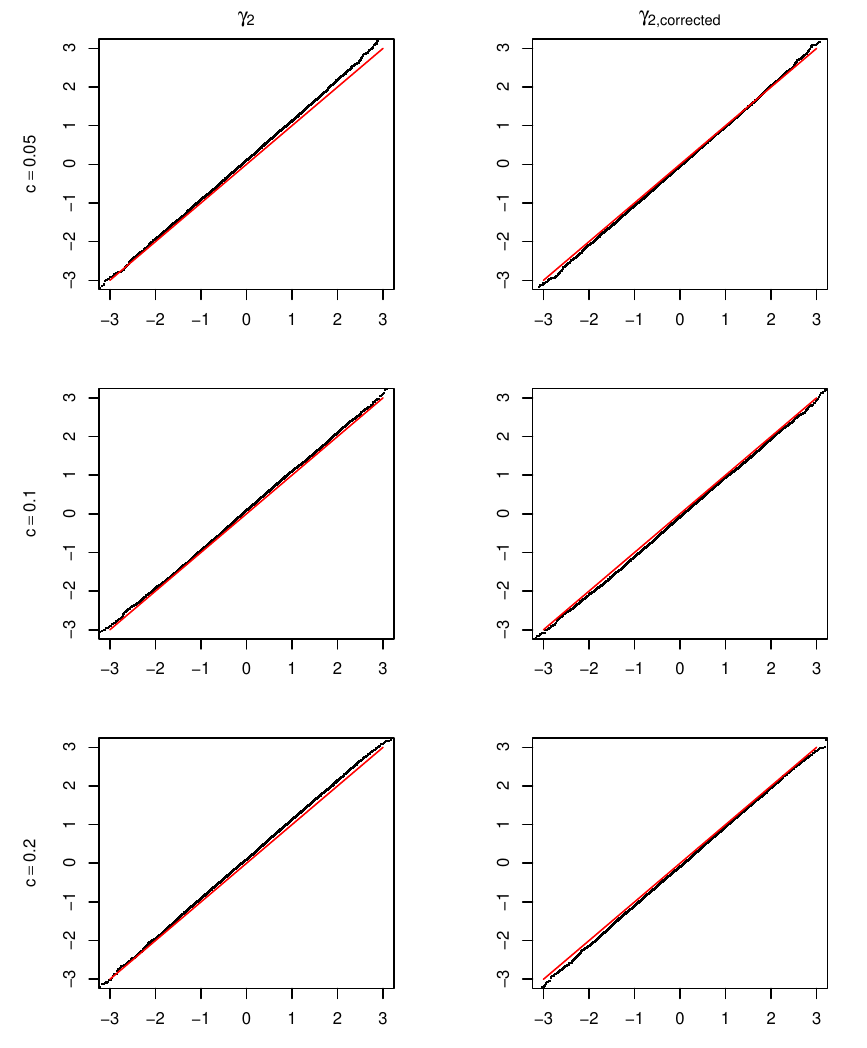} %
    \caption{The figures are the quantile-quatile plots of $(\hat{\gamma}_2 - \gamma_2)/\hat{\sigma}_2$
     against the standard normality under the setting Case II ($z_{ij1}\in \{-1,1\}$ with probabilities $0.3$ and $0.7$,
$z_{ij2}\in \{-1,1\}$ with probabilities $0.2$ and $0.8$),
    where the horizontal and vertical coordinates denote the theoretical
    and sample quantiles, respectively.
    Here, the number of items is $n=30$.
    The first, second, third rows correspond to $c=0.05$, $c=0.1$ and $c=0.2$ in the $\beta$-parameter settings.
     The plot under the case of $c=0$ is similar to that under $c=0.05$ and is not shown to save space.
     The first and second columns represent the uncorrected and corrected statistics for $\gamma_2$, where $\hat{\gamma}_2$ denotes the MLE of the standard error of $\hat{\gamma}_2$. } %
    \label{fig-non-30-gamma-2}
\end{figure}

\begin{table}[!htpb]\centering
\caption{The reported values are the coverage frequency ($\times 100\%$) for $\widehat{\bs{\gamma}}$,
the coverage frequency ($\times 100\%$) for $\widehat{\bs{\gamma}}_{bc}$ and the length of confidence interval, denoted as
``CF", $\mathrm{CF}_{corr}$ and ``LCI", respectively.}
\label{tab-small-gamma}
\footnotesize
\vspace{0.2cm}
\begin{tabular}{ c c ccc c ccc c ccc c ccc}
\hline
 &&\multicolumn{3}{c}{$c=0$}
&&\multicolumn{3}{c}{$c=0.05$}
&&\multicolumn{3}{c}{$c=0.1$}
&&\multicolumn{3}{c}{$c=0.2$}
\\
\cmidrule(r){3-5} \cmidrule(r){7-9} \cmidrule(r){11-13} \cmidrule(r){15-17}
  &&  CF&$\mathrm{CF}_{corr}$&LCI  && CF&$\mathrm{CF}_{corr}$&LCI
&&  CF&$\mathrm{CF}_{corr}$&LCI  && CF&$\mathrm{CF }_{corr}$&LCI
\\
\midrule
&&\multicolumn{15}{c}{Case I: $z_{ijk,1}\in \{-1,1\}$ with equal probability, $z_{ijk,2}\sim N(0,1)$}
\\
\cline{3-17}
$\gamma_1$ && $ 94.1 $&$ 94.1 $&$ 0.19 $&&$ 94.3 $&$ 93.7 $&$ 0.19 $&&$ 92.9 $&$ 93.8 $&$ 0.19 $&&$ 94.1 $&$ 94.9 $&$ 0.19 $
\\
$\gamma_2$ &&$ 94.1 $&$ 93.3 $&$ 0.19 $&&$ 94.3 $&$ 94.2 $&$ 0.19 $&&$ 92.9 $&$ 93 $&$ 0.19 $&&$ 94.1 $&$ 93.5 $&$ 0.19 $
\\
\midrule
&&\multicolumn{15}{c}{Case II: $z_{ijk,1}\in \{-1,1\}$ with $(0.3,0.7)$,
$z_{ijk,2}\in \{-1,1\}$ with $(0.2,0.8)$}
\\
\cline{3-17}
$\gamma_1$ &&$ 94.13 $&$ 94.33 $&$ 0.2 $&&$ 94.3 $&$ 94.38 $&$ 0.2 $&&$ 93.98 $&$ 94.04 $&$ 0.2 $&&$ 94.44 $&$ 94.34 $&$ 0.2 $
\\
$\gamma_2$ &&$ 94.43 $&$ 94.55 $&$ 0.2 $&&$ 94.28 $&$ 94.25 $&$ 0.2 $&&$ 94.77 $&$ 94.67 $&$ 0.2 $&&$ 94.46 $&$ 94.32 $&$ 0.2 $
\\
\bottomrule
\end{tabular}
\end{table}

\subsubsection{Diverging $p_n=30$, moderate $n=100$, and $m_{ij=1}$}

Here, we consider a large $p_n$.
A $p_n$-dimensional covariate vector $\bm{z}_{ijk}$ is generated as follows.
The first element of $\bm{z}_{ijk}$ takes one value of $\{-1,1\}$ with respective probabilities
$0.3$ and $0.7$.
The second to fourth elements of $\bm{z}_{ij}$ are generated from a $3$-dimensional normal distribution
with mean zero and a covariance matrix, whose diagonal elements are $1$'s and non-diagonal elements
are $0.3$'s. The other elements are independently generated from the standard normality.
Set $\bm{\gamma}=(0.3,0.3,0.3,0.3, 1/n, \ldots, 1/n)$, $n=100$ and $p_n=30$. Other parameters are the same as those in section 5.1.1.

The $95\%$ coverage frequency ($\times 100\%$) for $\beta_i-\beta_j$ is presented in
Table \ref{tab-high-beta-100}.
From this table, we can see that all simulated coverage frequencies are close to the target level.
The QQ plots of $\hat{\eta}_i$ against the standard normality are similar to those in Figure \ref{fig-non-100-beta-i}, where the QQ curve aligns with the diagonal line.
Therefore, they are not shown. There is similar reason,
the QQ plots for the $\gamma$ parameters, which are similar to those in
Figure \ref{fig-uniform-100-gamma}, are also not shown.

\begin{table}[!htpb]\centering
\caption{The reported values are the coverage frequency ($\times 100\%$) for $\beta_i-\beta_j$ for a pair $(i,j)$, denoted as ``CF" and the length of the confidence interval ``LCI", where the number of items is $100$.}
\label{tab-high-beta-100}
\footnotesize
\vspace{0.2cm}
\renewcommand{\arraystretch}{1.2}
\begin{tabular}{ c c cc c cc c cc c cc}
\hline
 &&\multicolumn{2}{c}{$c=0$}
&&\multicolumn{2}{c}{$c=0.05$}
&&\multicolumn{2}{c}{$c=0.1$}
&&\multicolumn{2}{c}{$c=0.2$}
\\
\cmidrule(r){3-4} \cmidrule(r){6-7} \cmidrule(r){9-10} \cmidrule(r){12-13}
  &&  CF&LCI  && CF&LCI
&&  CF&LCI  && CF&LCI
\\
\cline{3-13}
$(1,2)$ && $ 94.66 $&$ 1.19 $&&$ 95.04 $&$ 1.19 $&&$ 94.02 $&$ 1.19 $&&$ 95.16 $&$ 1.21 $
\\
$(49,50)$ &&$ 94.66 $&$ 1.19 $&&$ 94.8 $&$ 1.19 $&&$ 94.34 $&$ 1.19 $&&$ 94.46 $&$ 1.19 $
\\
$(99,100)$ &&$ 94.84 $&$ 1.19 $&&$ 94.66 $&$ 1.19 $&&$ 94.78 $&$ 1.19 $&&$ 94.72 $&$ 1.21 $
\\
$(0,1)$ &&$ 95.04 $&$ 1.19 $&&$ 95.08 $&$ 1.19 $&&$ 94.58 $&$ 1.19 $&&$ 94.86 $&$ 1.21 $
\\
$(0,50)$ &&$ 95.06 $&$ 1.19 $&&$ 94.94 $&$ 1.19 $&&$ 94.96 $&$ 1.19 $&&$ 94.6 $&$ 1.2 $
\\
$(0, 100)$ &&$ 94.4 $&$ 1.19 $&&$ 94.46 $&$ 1.19 $&&$ 94.52 $&$ 1.19 $&&$ 94.8 $&$ 1.21 $
\\
\bottomrule
\end{tabular}
\end{table}

Next, we consider Erd\H{o}s--R\'{e}nyi comparison graphs, where each pair is compared with probability
$q$. Set $\beta_i = (i-1)*0.2*\log n/(n-1)$.
Other settings are the same as those in the above simulation.
The $95\%$ coverage frequencies for $\beta_i-\beta_j$ are reported in Table \ref{tab-ER-beta}.
As expected, the estimation accuracy under sparse comparison designs is worse than that under
dense comparison designs. When $q$ becomes smaller, the length of confidence interval becomes larger.
When $q=0.1$, the coverage frequencies are far lower than the target level $95\%$.
There are similar phenomena for the MLE of the covariate parameter, whose results are not shown.

\begin{table}[!htpb]\centering
\caption{The reported values are the coverage frequency ($\times 100\%$) for $\beta_i-\beta_j$ for a pair $(i,j)$, denoted as ``CF" and the length of the confidence interval ``LCI", where the number of items is $100$ under Erd\H{o}s--R\'{e}nyi comparison graph.}
\label{tab-ER-beta}
\footnotesize
\vspace{0.2cm}
\renewcommand{\arraystretch}{1.2}
\begin{tabular}{ c c cc c cc c cc c cc}
\hline
 &&\multicolumn{2}{c}{$q=0.1$}
&&\multicolumn{2}{c}{$q=0.3$}
&&\multicolumn{2}{c}{$q=0.5$}
&&\multicolumn{2}{c}{$q=0.7$}
\\
\cmidrule(r){3-4} \cmidrule(r){6-7} \cmidrule(r){9-10} \cmidrule(r){12-13}
  &&  CF&LCI  && CF&LCI
&&  CF&LCI  && CF&LCI
\\
\cline{3-13}
$(1,2)$ && $ 88.14 ( 3.58 ) $&$ 4.7 $&&$ 92.92 $&$ 2.31 $&&$ 95.16 $&$ 1.74 $&&$ 94.82 $&$ 1.45 $
\\
$(49,50)$ &&$ 87.93 ( 2.88 ) $&$ 4.66 $&&$ 93.44 $&$ 2.28 $&&$ 94.34 $&$ 1.72 $&&$ 95.18 $&$ 1.43 $
\\
$(99,100)$ &&$ 87.96 ( 4.02 ) $&$ 4.73 $&&$ 93.86 $&$ 2.31 $&&$ 94.46 $&$ 1.74 $&&$ 94.44 $&$ 1.46 $
\\
$(0,1)$ &&$ 88.75 ( 3.46 ) $&$ 4.71 $&&$ 93.24 $&$ 2.31 $&&$ 94.74 $&$ 1.74 $&&$ 94.52 $&$ 1.45 $
\\
$(0,50)$ &&$ 88.45 ( 3.2 ) $&$ 4.68 $&&$ 93.86 $&$ 2.29 $&&$ 94.92 $&$ 1.73 $&&$ 94.52 $&$ 1.44 $
\\
$(0, 100)$ &&$ 86.89 ( 3.46 ) $&$ 4.73 $&&$ 93.56 $&$ 2.31 $&&$ 94.2 $&$ 1.74 $&&$ 94.8 $&$ 1.46 $
\\
\bottomrule
\end{tabular}
\end{table}

\subsubsection{Influence of missing covariates}

This section evaluates the influence of missing covariates on statistical inference.
All simulation situation settings are the same as those in  section \ref{subsubsec-moderate}.
Therefore, the correct model has two covariates. To assess the influence of missing covariates on statistical inference,
we fitted the data under three different models:
the correct model with two covariates,  the miss-specified model with only one covariate $z_{ijk,1}$
and the  Bradley--Terry model without covariates.
Since there are similar phenomena under the covariate setting (Case I) and (Case II),
we only select some QQ-plots under Case II here.
Consider the covariate setting Case II in section \ref{subsubsec-moderate}

The QQ-plots for $(\hat{\beta}_i - \beta_i)/\hat{\sigma}_i$ with three different indices $i=1, n/2, n$ against the standard normality are in Figure \ref{fig-comparison-beta}.
From this figure, we can see that (1) if $\beta_i\neq 0$,  the sample quantile agrees with the quantile of the standard normality only when the model is correct;
(2) the QQ-plots lie above the diagonal line under the miss-specified model with only one covariate, and the  Bradley--Terry model without covariates, indicating
a positive bias for the estimator,
(3) by comparing the plots between the second row and the third row,
the straight-line distance  between  the QQ-curve and the diagonal line becomes larger when the number of missed covariates increases;
this phenomenon is also observed by comparing the plots between the second and third columns in the same row, where the values of $\beta_i$ increases.
There are similar phenomena  for the QQ-plots of $\hat{\gamma}_1$, which are not shown to save space. 

\begin{figure}[!htbp] 
    \centering
    \includegraphics[width=0.8\textwidth]{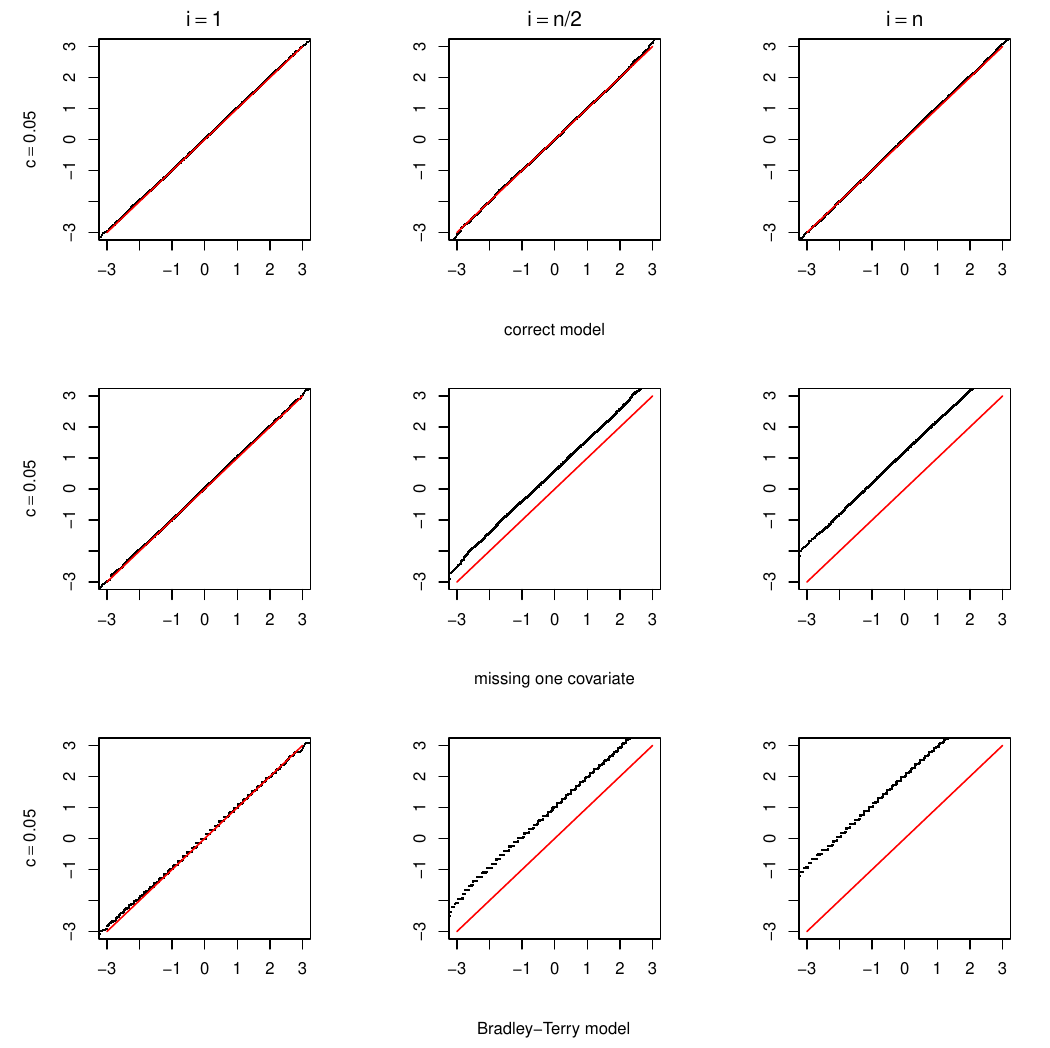} %
    \caption{The figures are the quantile-quatile plots of $(\hat{\beta}_i - \beta_i)/\hat{\sigma}_i$
     against the standard normality under the setting Case II ($z_{ij1}\in \{-1,1\}$ with probabilities $0.3$ and $0.7$,
$z_{ij2}\in \{-1,1\}$ with probabilities $0.2$ and $0.8$),
    where the horizontal and vertical coordinates denote the theoretical
    and sample quantiles, respectively.
    Here, the number of items is $n=100$ and each pair is compared once.
    the $\beta$-parameter setting is $\beta_i = c i\log n/n$ for $i=0, \ldots, n$. Here, I only show the plots when $c=0.05$, where there are similar phenomena when $c=0, 0.1, 0.2$.
     The first, second and third rows correspond to the correct model with two covariates, the miss-specified model with only one covariate $z_{ij1}$,
     and the original Bradley--Terry model without covariates.
      } %
    \label{fig-comparison-beta}
\end{figure}

In order to compare which of two ranking models is better to fit the data,
 we further propose a novel rank-goodness-of-fit measure to compare two models when applying them into paired comparison data.
Let $\hat{p}^{(model~A)}_{ijk}$ and $\hat{p}^{(model~B)}_{ijk}$ denote the respective estimator of the probability of item $i$ winning item $j$ in the $k$-th comparison under
model A and model B.
The rank-goodness-of-fit measure is defined as
\begin{small}
\begin{equation}
\label{definition-RAB}
\mathrm{R}_{A,B}= \frac{1}{N} \sum_{i<j}\sum_k \left( 1\{ a_{ijk}=1, \hat{p}^{(model~A)}_{ijk} > \hat{p}^{(model~B)}_{ijk} \} +  1\{ a_{ijk}=0, \hat{p}^{(model~A)}_{ijk} < \hat{p}^{(model~B)}_{ijk} \} \right),
\end{equation}
\end{small}
where $1\{\cdot\}$ denotes an indicator function and $N$ is the total number of paired comparisons.
This is motivated by the home-field effect, where the probability of a home team winning an away team in the CBTM with the home/away covariate
is higher than that in the original Bradley--Terry model.
For this measure, we have the relational expression $\mathrm{R}_{A,B} + \mathrm{R}_{B,A}=1$.
The averaged rank-goodness-of-fit measure is reported in Table \ref{tab-rgf},
where two miss-specified models with only one covariate and no covariate are compared with the correct model with two covariates.
From this table, we can see the proposed measure is much higher than $1/2$ under the correct model.

\begin{table}[!htpb]\centering
\caption{The average rank-goodness-of-fit measure ($\times 100\%$). $R_{corr,miss}$ denotes the comparison between the correct model and the model with one missing covariate.
$R_{corr,BT}$ denotes the comparison between the correct model and the Bradley--Terry model with no covariates.}
\label{tab-rgf}
\footnotesize
\vspace{0.2cm}
\renewcommand{\arraystretch}{1.2}
\begin{tabular}{  cc c cc c cc c cc}
\hline
 \multicolumn{2}{c}{$c=0$}
&&\multicolumn{2}{c}{$c=0.05$}
&&\multicolumn{2}{c}{$c=0.1$}
&&\multicolumn{2}{c}{$c=0.2$}
\\
\cmidrule(r){1-2} \cmidrule(r){4-5} \cmidrule(r){7-8} \cmidrule(r){10-11}
  $R_{corr,miss}$ & $R_{corr,BT}$  &&  $R_{corr,miss}$& $R_{corr,BT}$
&&  $R_{corr,miss}$&$R_{corr,BT}$  && $R_{corr,miss}$&$R_{corr,BT}$
\\
\midrule
 $ 58.33 $ & $59.13 $&&$ 59.16 $ & $59.28 $&&$ 59.95 $ &  $59.41 $&&$ 61.33 $  & $59.53 $
\\
\bottomrule
\end{tabular}
\end{table}

}}

\subsection{A real data example}

The National Basketball Association (NBA) is the world's premier men's professional basketball league and is one of the major professional sports leagues in North America.
It contains 30 teams equally divided into Eastern and Western conferences.
In the regular season, each team plays two, three, or four games against another, for a total of 82 games, of which 41 games
were in their home arena, and 41 were played away. Thus, there were 1, 230 games in the NBA regular season.
We used the 2018-19 NBA regular season data as an example, which is available from
\url{https://www.landofbasketball.com/yearbyyear/2018_2019_teams.htm}.
We consider ``home/away" as the covariate. When team $i$ interacts with team $j$ in the $k$th comparison,
we set $z_{ijk}=1$ if $i$ is at home; otherwise $z_{ijk}=-1$.
The fitted merits are given in Table \ref{Table:beta:real}, where we use ``Washington Wizards" as
the baseline ($\beta_n = 0$).

The estimated home effect $\widehat{\gamma}$ and its standard errors are $0.45$ and $0.065$, respectively.
Under the null hypothesis of having no home effects, this gives a $p$-value $2.1\times 10^{-12}$, indicating a significant home advantage.
{\color{black}{Since a natural way for the goodness-of-fit test of the classical Bradley--Terry model is to embed it in a larger family of distributions--the CBTM,
where the corresponding task is to test the null $\gamma=0$,
the p-value indicates that it rejects the  Bradley--Terry model to fit this data set under the usual nominal level $0.05$.
In addition, when applying the proposed rank-goodness-of-fit measure in \eqref{definition-RAB} to compare the CBTM and the Bradley--Terry model in this season data,
it is $0.59$, which is much larger than $0.5$.
When applying  this measure to the postseason data after the 2018-19 NBA regular season, where the fitted parameters from  the regular season data are used,
the value is $0.56$ out of $82$ games, which is also larger than $0.5$. This indicates that
the CBTM has a better fitness than the classical Bradley--Terry model.
}}

It would be interesting to compare the order of the eight playoff seeds at the two conferences
by the NBA rule, with ordering based on the merits obtained from fitting the
Bradley--Terry model. The order from high to low in the eight playoff seeds of the Western conference is as follows:
Warriors, Nuggets, Trailblazers, Rockets, Glasses, Thunder, Spurs, and Clippers.
The corresponding order at the Eastern Conference was:
Bucks, Raptors, 76ers, Celtics, Pacers, Nets, Magics, Pistons.
From Table \ref{Table:beta:real}, we
see that the ordering of the merits of the Eastern Conference is consistent with that of
the NBA rule. In addition, at the Western conference, the order of seven and eight seeds was switched.

{\renewcommand{\arraystretch}{1}
\begin{table}[!hbt]\centering
\scriptsize
\caption{The estimates of $\beta_i$ and their standard errors in 2018-19 NBA regular season.}
\label{Table:beta:real}
\begin{tabular}{l lccc c lccc }
\hline
Order &Subject & $d_i$  &  $\hat{\beta}_i$ & $\hat{\sigma}_i(\times 10)$ & & Subject & $d_i$  &  $\hat{\beta}_i$ & $\hat{\sigma}_i(\times 10)$ \\
\hline
  & \multicolumn{4}{c}{Western conference} & & \multicolumn{4}{c}{Eastern conference} \\
   \cline{2-5} \cline{7-10}
$ 1 $& Golden State Warriors &$ 57 $&$ 1.5 $&$ 3.52 $& & Milwaukee Bucks &$ 60 $&$ 1.6 $&$ 3.59 $\\
$ 2 $& Denver Nuggets &$ 54 $&$ 1.34 $&$ 3.47 $& & Toronto Raptors &$ 58 $&$ 1.48 $&$ 3.54 $\\
$ 3 $& Portland Trail Blazers &$ 53 $&$ 1.28 $&$ 3.46 $& & Philadelphia 76ers &$ 51 $&$ 1.07 $&$ 3.45 $\\
$ 4 $& Houston Rockets &$ 53 $&$ 1.27 $&$ 3.46 $& & Boston Celtics &$ 49 $&$ 0.95 $&$ 3.43 $\\
$ 5 $& Utah Jazz &$ 50 $&$ 1.09 $&$ 3.43 $& & Indiana Pacers &$ 48 $&$ 0.89 $&$ 3.43 $\\
$ 6 $& Oklahoma City Thunder &$ 49 $&$ 1.04 $&$ 3.43 $& & Brooklyn Nets &$ 42 $&$ 0.57 $&$ 3.41 $\\
$ 7 $& Los Angeles Clippers &$ 48 $&$ 0.98 $&$ 3.41 $& & Orlando Magic &$ 42 $&$ 0.57 $&$ 3.41 $\\
$ 8 $& San Antonio Spurs &$ 48 $&$ 0.97 $&$ 3.41 $& & Detroit Pistons &$ 41 $&$ 0.52 $&$ 3.41 $\\
$ 9 $& Sacramento Kings &$ 39 $&$ 0.49 $&$ 3.4 $& & Miami Heat &$ 39 $&$ 0.42 $&$ 3.4 $\\
$ 10 $& Los Angeles Lakers &$ 37 $&$ 0.4 $&$ 3.41 $& & Charlotte Hornets &$ 39 $&$ 0.42 $&$ 3.41 $\\
$ 11 $& Minnesota Timberwolves &$ 36 $&$ 0.36 $&$ 3.4 $& & Washington Wizards &$ 32 $&$ 0 $&$ 3.45 $\\
$ 12 $& Memphis Grizzlies &$ 33 $&$ 0.19 $&$ 3.43 $& & Atlanta Hawks &$ 29 $&$ -0.14 $&$ 3.47 $\\
$ 13 $& Dallas Mavericks &$ 33 $&$ 0.18 $&$ 3.43 $& & Chicago Bulls &$ 22 $&$ -0.56 $&$ 3.6 $\\
$ 14 $& New Orleans Pelicans &$ 33 $&$ 0.15 $&$ 3.43 $& & Cleveland Cavaliers &$ 19 $&$ -0.77 $&$ 3.69 $\\
$ 15 $& Phoenix Suns &$ 19 $&$ -0.69 $&$ 3.68 $& & New York Knicks &$ 17 $&$ -0.9 $&$ 3.76 $\\
\hline
\end{tabular}

\end{table}

\section{Summary and discussion}
\label{section:sd}

We present the maximum likelihood  estimation of the CBTM.
{\color{black}{
Using a two-stage process, we demonstrated the consistency of the MLE when the number of subjects approached infinity.
Furthermore, by applying a two/third-order Taylor expansion to
score functions, we establish the asymptotic normality of the MLE.
The strategies for deriving the asymptotic properties of the MLE
shed light on a principal approach to similar problems. These principles should apply to
a class
of paired comparison models, in which the logistic distribution in the Bradley--Terry model is replaced by
using other distributions such as the probit distribution in the Thurstone model [\cite{thurstone1994a}]
even for generalized Bradley--Terry models with ties [\cite{davidson1970on,rao1967ties}].
}}

{\color{black}{
There exist several issues that need to further investigate.
First, the consistency result in Theorem \ref{Theorem:con} restricts the diverging rate of the model parameters $\bm{\beta}$ and $\bm{\gamma}$
while the asymptotic distributions in Theorems \ref{theorem-central-b} and \ref{Theorem-central-a} require a stronger condition.
The simulation results indicate these conditions might be relaxed.
Note that the asymptotic behavior of the MLE depends on the configuration of all parameters.
It would be interesting to investigate whether these conditions could be relaxed.
Second, we establish asymptotic distributions of the MLE in Theorems \ref{Theorem-central-a} and \ref{theorem-central-b} under the condition that
the dimension of covariates is fixed,
and do not study the asymptotic distributions of the MLE when the dimension of covariates diverge.
In addition, this problem is also not studied in case of Erd\H{o}s--R\'{e}nyi comparison graphs.
In the latter,  the dimension of covariates could influence the convergence rate of the covaraite parameter estimator $\bm{\hat{\gamma}}$,
which further affects the remainder terms in the asymptotic representation of the MLE $\bm{\hat{\beta}}$ of the merit parameter.
The simulation results indicate that the asymptotic normal approximation of the MLE is good.
It would also be of interest to explore this problem.
Third, although we introduce a bias-corrected estimator for the regression parameter of covariates in Remark \ref{remark-gamma},
we do not give a formal result on its asymptotic distribution.
To obtain its asymptotic distribution of $\widehat{\bs{\gamma}}_{bc}$, it faces two difficulties:
(1) how to derive the asymptotic variance of $\widehat{B}$; (2) how to analyze
the correlation between $\widehat{\bs{\gamma}}$ and $\widehat{B}$.
It is beyond of the current manuscript to study this problem. We would like to investigate it in the future work.
}}

\section*{Acknowledgements}
The supplementary material and simulation codes are available upon request by sending an email to \textit{tingyanty@mail.ccnu.edu.cn}.

\setlength{\itemsep}{-1.5pt}
\setlength{\bibsep}{0ex}
\bibliography{reference3}
\bibliographystyle{apa}

\end{document}